\documentclass[submission,copyright,creativecommons]{eptcs}
\providecommand{\event}{TLLA 2020}
\usepackage{underscore}

\usepackage{amsmath}
\usepackage{amsthm}
\usepackage{amsfonts}
\usepackage{amssymb}
\usepackage{bussproofs}
\usepackage{cmll}
\usepackage{graphicx}
\usepackage{stmaryrd}

\newcounter{myequation}
\makeatletter
\@addtoreset{equation}{myequation}
\makeatother

\newcommand*{\ie}{\emph{i.e.}}
\newcommand*{\Nat}{\ensuremath{\mathbb{N}}}
\newcommand*{\A}{\Gamma}
\newcommand*{\B}{\Delta}
\newcommand*{\C}{\Sigma}

\newcommand*{\Bool}{\ensuremath{\mathbb{B}}}
\newcommand*{\true}{\textnormal{\textsf{true}}}
\newcommand*{\false}{\textnormal{\textsf{false}}}

\newcommand*{\yalla}{\textnormal{\textsf{Yalla}}}
\newcommand*{\Coq}{\textnormal{\textsf{Coq}}}

\newcommand*{\PTIME}{\textsf{PTIME}}

\newcommand*{\size}[1]{\left|#1\right|}
\newcommand*{\sgte}{{\ensuremath{\bullet}}}
\newcommand*{\epg}{{\ensuremath{\star}}}
\newcommand*{\e}{{\ensuremath{e}}}

\DeclareSymbolFont{epsilon}{OML}{ntxmi}{m}{it}
\DeclareMathSymbol{\epsilon}{\mathord}{epsilon}{"0F}
\newcommand*{\eb}{\ensuremath{\epsilon}}
\newcommand*{\et}{\ensuremath{\varepsilon}}
\newcommand*{\one}{1}
\newcommand*{\zero}{0}
\newcommand*{\tens}{\otimes}
\newcommand*{\plus}{\oplus}
\newcommand*{\orth}{^\bot}
\newcommand*{\pg}{\S}
\newcommand*{\copg}{\bar\pg}
\newcommand*{\ax}{\ensuremath{\text{ax}}}
\newcommand*{\cut}{\ensuremath{\text{cut}}}
\newcommand*{\exc}{\ensuremath{\text{ex}}}
\newcommand*{\llde}[1][]{\ensuremath{\wn_{#1}\textit{d}}}
\newcommand*{\llwk}[1][]{\ensuremath{\wn_{#1}\textit{w}}}
\newcommand*{\llco}[1][]{\ensuremath{\wn_{#1}\textit{c}}}
\newcommand*{\lldg}{\ensuremath{\textnormal{??}}}
\newcommand*{\de}{\depred}
\newcommand*{\co}[1]{\copred\ensuremath{_{#1}}}
\newcommand*{\dg}{\dgpred}
\newcommand*{\prom}{\prompred}
\newcommand*{\promg}{\prompred_g}
\newcommand*{\promleq}{\prompred_{\leq}}

\newcommand*{\logsys}[1]{\textnormal{\textsf{#1}}}
\newcommand*{\LL}{\logsys{LL}}

\newcommand*{\ELL}{\logsys{ELL}}
\newcommand*{\LLL}{\logsys{LLL}}
\newcommand*{\SLL}{\logsys{SLL}}
\newcommand*{\BLL}{\logsys{BLL}}
\newcommand*{\BsLL}{\logsys{B\ensuremath{_S}LL}}
\newcommand*{\seLL}{\logsys{seLL}}
\newcommand{\AIC}[1]{\AxiomC{\ensuremath{#1}}}
\newcommand{\ZIC}[1]{\AxiomC{}\UnaryInfC{\ensuremath{#1}}}
\newcommand{\UIC}[1]{\UnaryInfC{\ensuremath{#1}}}
\newcommand{\BIC}[1]{\BinaryInfC{\ensuremath{#1}}}
\newcommand{\TIC}[1]{\TrinaryInfC{\ensuremath{#1}}}
\newcommand{\QIC}[1]{\QuaternaryInfC{\ensuremath{#1}}}
\newcommand{\RL}[1]{\RightLabel{\ensuremath{#1}}}
\newcommand{\DP}{\DisplayProof}

\newtheorem{thm}{Theorem}
\newtheorem{lem}{Lemma}
\newtheorem{prop}{Proposition}

\newtheorem*{thmst}{Theorem}
\theoremstyle{definition}
\newtheorem{nota}{Notation}
\newtheorem{defi}{Definition}
\theoremstyle{remark}
\newtheorem{rem}{Remark}
\newtheorem{ex}{Example}

\newcommand*{\sset}{\ensuremath{\mathcal{E}}}
\newcommand*{\mpx}[1]{\ensuremath{\wn{\textit{m}_{#1}}}}
\newcommand*{\permof}[2]{#2\cdot#1}
\newcommand*{\superLL}{\logsys{superLL}}
\newcommand*{\ocf}{\oc_f}
\newcommand*{\ocu}{\oc_u}
\newcommand*{\depred}{\textsc{de}}
\newcommand*{\copred}{\textsc{co}}
\newcommand*{\dgpred}{\textsc{dg}}
\newcommand*{\prompred}{\textsc{p}}

\newcommand*{\leqse}{\preccurlyeq}

\newcommand*{\reduc}{\rightsquigarrow}
\newcommand*{\transfo}{\mapsto}

\title{Super Exponentials in Linear Logic\thanks{This work was supported by the IRN Linear Logic, and by the LABEX MILYON (ANR-10-LABX-0070) of Universit\'e de Lyon, within the program ``Investissements d’Avenir'' (ANR-11-IDEX-0007) operated by the French National Research Agency (ANR).}}
\author{Esa\"ie Bauer \quad\qquad\qquad\qquad Olivier Laurent\institute{Univ Lyon, EnsL, UCBL, CNRS,  LIP\\ F-69342, LYON Cedex 07, France}\email{esaie.bauer@irif.fr \qquad\qquad olivier.laurent@ens-lyon.fr}}
\def\titlerunning{Super Exponentials}
\def\authorrunning{E. Bauer \& O. Laurent}

\begin{document}
\maketitle

\begin{abstract}
  Following the idea of Subexponential Linear Logic and Stratified Bounded Linear Logic, we propose a new parameterized version of Linear Logic which subsumes other systems like \ELL, \LLL{} or \SLL, by including variants of the exponential rules.
  We call this system Superexponential Linear Logic (\superLL).
  Assuming some appropriate constraints on the parameters of \superLL, we give a generic proof of cut elimination.
  This implies that each variant of Linear Logic which appears as a valid instance of \superLL{} also satisfies cut elimination.
\end{abstract}

Linear logic (\LL) has been introduced by Jean-Yves Girard in 1987~\cite{ll}.
Since then, it has become a pervasive tool in proof theory, in typing systems and semantics for programming languages, in computational complexity theory, etc.
The key property which provides a computational meaning to this logic is cut elimination.

During the years, many variants of \LL{} have been introduced which differ in particular on some specific uses of exponential rules.
Each time a dedicated proof of cut elimination is provided by the authors.
We are interested in finding a generic cut-elimination proof for as many systems as possible.

Proving the cut-elimination theorem for many systems at once is already the idea behind the parametric system of Subexponential Linear Logic (\seLL)~\cite{structexp,subexp}.
However it relies on a parameterized version of Girard's promotion rule, and thus rules out systems based on other kinds of promotions such as functorial promotion.
Parameters of \seLL{} allow to control $\wn$-rules.
Exponential connectives are indexed by some exponential signatures (instead of a single pair $\{\oc,\wn\}$).
These signatures are equipped with a pre-order structure used in extending Girard's promotion rule.
Some closure properties of the parameters (with respect to the pre-order) are required for cut elimination to hold.
The idea of indexing the exponential modalities is also at the heart of Stratified Bounded Linear Logic (\BsLL)~\cite{llcoeffects}. Indexing is there based on a semi-ring endowed with a compatible partial order.

The new system we consider is called Superexponential Linear Logic (\superLL).
Its $\wn$-rules are parameterized by predicates which provide the valid relations between the exponential signatures used in the premises and in the conclusion of each rule.
In order to take into account variants of \LL{} used in implicit computational complexity (\ELL~\cite{lll}, \LLL~\cite{lll}, \SLL~\cite{sll}), it is simpler to consider a system based on a functorial version of promotion together with an explicit digging rule.
As a counter part, we have to understand how this is related with Girard's promotion rule.

Under appropriate axioms on the parameters, we can describe various proof transformations on \superLL{} including in particular cut elimination.
Choosing specific instances of \superLL{} leads to systems equivalent to a number of variants of \LL{} from the literature (some light systems for complexity, but also \seLL{} or \BsLL).

\bigskip

In Sections~\ref{secll} and~\ref{secmiscll}, we recall the definitions of \LL{} and of the variants we are going to consider.
The notion of \sset-formula which deals with indexed exponential connectives is introduced.
Section~\ref{secsuperll} contains the formal definition of the rules of \superLL.
Section~\ref{secce} is the core part of the paper: it contains the proof of the cut-elimination property for \superLL.
After describing the proof sketch (Section~\ref{seccesketch}) which pinpoints the requirements on the parameters, we give the list of axioms we rely on (Section~\ref{secceax}).
These axioms are the crucial ingredients of the substitution lemma (Section~\ref{secsubstlem}) which allows us to eliminate cuts on exponential formulas.
Section~\ref{sectransfo} presents other proof transformations required to move from one presentation of a system to another.
Based on appropriate axioms, it is shown how to introduce the Girard's style promotion rule, or an ordered version of this rule similar to \seLL's promotion.
Finally Section~\ref{secsubsyst} describes how to define the systems of Sections~\ref{secll} and~\ref{secmiscll} as instances of \superLL{} which satisfy the axioms of Section~\ref{secceax} and how to deduce cut elimination from the generic proof of Section~\ref{secce}.

\section{Linear Logic}\label{secll}

In order to cover the various systems under consideration in this paper, we define a generalization of \LL{} formulas with an indexed family of exponential connectives.

\begin{defi}[Linear \sset-Formulas]
Given a set \sset{}, \emph{(linear) \sset-formulas} are generated by:
\begin{equation*}
A ::= X \mid X\orth \mid A\tens A \mid A\parr A \mid \one \mid \bot \mid A\with A \mid A\plus A \mid \top \mid \zero \mid \oc_\e A \mid \wn_\e A
\qquad\text{where $\e\in\sset$.}
\end{equation*}
\end{defi}

\begin{nota}
  Elements of \sset{} are called \emph{exponential signatures}.
  If $\vec{e}=\e_1,\dotsc,\e_n$, we use the notation $\wn_{\vec{\e}}A$ for $\wn_{\e_1}\dotsc\wn_{\e_n}A$.
\end{nota}

Usual \LL{} formulas correspond to the particular case where \sset{} is a singleton set (let say $\sset=\{\sgte\}$).
In this case we simply use the notations $\oc A:=\oc_\sgte A$ and $\wn A:=\wn_\sgte A$.

As usual a \emph{duality} operation $A\mapsto A\orth$ is defined on all \sset-formulas (not just for $X\orth$).
It is the involution satisfying:
\begin{equation*}
  \begin{array}{rcl@{\qquad\qquad}rcl@{\qquad\qquad}rcl}
    (A\tens B)\orth&=&A\orth\parr B\orth & \one\orth&=&\bot& && \\
    (A\parr B)\orth&=&A\orth\tens B\orth & \bot\orth&=&\one& (X\orth)\orth&=&X \\
    (A\with B)\orth&=&A\orth\plus B\orth & \top\orth&=&\zero& (\oc_\e A)\orth&=&\wn_\e A\orth \\
    (A\plus B)\orth&=&A\orth\with B\orth & \zero\orth&=&\top& (\wn_\e A)\orth&=&\oc_\e A\orth
  \end{array}
\end{equation*}
As often done in the literature, thanks to this duality, we focus on one-sided sequents for the sequent calculi under consideration.
Such a sequent is written $\vdash\A$ where $\A$ is a list of \sset-formulas.
The length of a list $\A$ is denoted $\size{\A}$.

Linear Logic (\LL) deals with formulas with only one kind of exponentials (\ie{} with formulas built from a singleton set $\sset=\{\sgte\}$).
Among the rules of \LL~\cite{ll} which are recalled in Table~\ref{tabllrules}, we distinguish between \emph{non-exponential rules} and \emph{exponential rules}.
Indeed the different systems under consideration will share the non-exponential ones and differ only on the exponential ones.

\begin{table}
\centering
\paragraph{Non-Exponential Rules}
\begin{gather*}
  \RL{\ax}
  \ZIC{\vdash A,A\orth}
  \DP
\qquad\qquad
  \AIC{\vdash A,\A}
  \AIC{\vdash A\orth,\B}
  \RL{\cut}
  \BIC{\vdash\A,\B}
  \DP
\qquad\qquad
  \AIC{\vdash\A}
  \RL{\exc}
  \UIC{\vdash\permof{\sigma}{\A}}
  \DP
\\[2ex]
  \AIC{\vdash A,\A}
  \AIC{\vdash B,\B}
  \RL{\tens}
  \BIC{\vdash A\tens B,\A,\B}
  \DP
\qquad\qquad
  \AIC{\vdash A,B,\A}
  \RL{\parr}
  \UIC{\vdash A\parr B,\A}
  \DP
\qquad\qquad
  \RL{\one}
  \ZIC{\vdash\one}
  \DP
\qquad\qquad
  \AIC{\vdash\A}
  \RL{\bot}
  \UIC{\vdash\bot,\A}
  \DP
\\[2ex]
  \AIC{\vdash A,\A}
  \AIC{\vdash B,\A}
  \RL{\with}
  \BIC{\vdash A\with B,\A}
  \DP
\qquad\qquad
  \AIC{\vdash A,\A}
  \RL{\plus_1}
  \UIC{\vdash A\plus B,\A}
  \DP
\qquad\qquad
  \AIC{\vdash B,\A}
  \RL{\plus_2}
  \UIC{\vdash A\plus B,\A}
  \DP
\qquad\qquad
  \RL{\top}
  \ZIC{\vdash\top,\A}
  \DP
\end{gather*}
\paragraph{Exponential Rules}
\begin{equation*}
  \AIC{\vdash A,\wn\A}
  \RL{\oc}
  \UIC{\vdash\oc A,\wn\A}
  \DP
\qquad\qquad
  \AIC{\vdash A,\A}
  \RL{\llde}
  \UIC{\vdash\wn A,\A}
  \DP
\qquad\qquad
  \AIC{\vdash\A}
  \RL{\llwk}
  \UIC{\vdash\wn A,\A}
  \DP
\qquad\qquad
  \AIC{\vdash\wn A,\wn A,\A}
  \RL{\llco}
  \UIC{\vdash\wn A,\A}
  \DP
\end{equation*}
\caption{Linear Logic Rules}\label{tabllrules}
\end{table}

In the ($\exc$) rule of Table~\ref{tabllrules}, if $\Gamma$ has length $n$, $\sigma$ is a permutation of $n$ elements and $\permof{\sigma}{\Gamma}$ denotes its action on $\Gamma$.
In the whole paper, we will deal with this exchange rule in an implicit manner.
This means that we will omit it in all discussions to make things lighter.
There are two ways of justifying this approach.
First, considering sequents as finite multi-sets rather than lists would exactly correspond to make exchange rules useless.
Second, all the mentioned results have been checked with explicit consideration of the exchange rules.

Concerning terminology, a ($\cut$) rule for which the cut formula $A$ has main connective $\oc_\e$ or $\wn_\e$ is called an \emph{exponential cut rule}. Other instances are called \emph{non-exponential cut rules}.
We call \emph{promotion rules} those introducing the $\oc$ connectives.
We call \emph{$\wn$-rules} the rules which introduce the $\wn$ connectives (independently of the $\oc$ connective), that is non-promotion exponential rules.
A rule is \emph{not acting} on a formula $A$ if $A$ is in the context of the rule and if the rule is not a promotion.

\begin{defi}[Derivability and Admissibility]
Let us consider a rule $R$:
\begin{gather*}
\AIC{\vdash \A_1}
\AIC{\dotsb}
\AIC{\vdash \A_n}
\RL{R}
\TIC{\vdash \A}
\DP
\end{gather*}
It is \emph{derivable} in a system $\mathcal{S}$, if there exists a proof tree which allows us to derive $\vdash \A$ from the sequents $\vdash\A_1$, \dots, $\vdash\A_n$ by using rules of $\mathcal{S}$.

It is \emph{admissible} in a system $\mathcal{S}$, if whenever $\vdash\A_1$, \dots, $\vdash\A_n$ are provable in $\mathcal{S}$, then $\vdash\A$ as well. So that derivable entails admissible, while the converse is not always true.

Two systems are said to be \emph{equivalent} if the provable sequents are the same, that is if all rules in one system are admissible in the other one, and conversely.
\end{defi}

\section{Other Linear Logic Systems}\label{secmiscll}

We present here different linear logic systems from the literature.
These systems differ only on their exponential rules.
They all deal with \sset-formulas (for an appropriate \sset) and one-sided sequents.

The first three systems below deal with $\{\sgte\}$-formulas (\ie{} with only one kind of exponentials).

\subsection{Functorial Promotion}\label{secfunprom}

\LL{} with functorial promotion is an alternative presentation of \LL{} particularly well suited for categorical semantics~\cite{asperti}.
It decomposes promotion into the so-called \emph{functorial promotion} and a new $\wn$-rule ($\lldg$) called \emph{digging}.
Its exponential rules are then:
\begin{equation*}
  \AIC{\vdash A,\A}
  \RL{\ocf}
  \UIC{\vdash \oc A,\wn\A}
  \DP
\qquad\quad
  \AIC{\vdash\wn\wn A,\A}
  \RL{\lldg}
  \UIC{\vdash\wn A,\A}
  \DP
\qquad\qquad
  \AIC{\vdash A,\A}
  \RL{\llde}
  \UIC{\vdash\wn A,\A}
  \DP
\qquad\qquad
  \AIC{\vdash\A}
  \RL{\llwk}
  \UIC{\vdash\wn A,\A}
  \DP
\qquad\qquad
  \AIC{\vdash\wn A,\wn A,\A}
  \RL{\llco}
  \UIC{\vdash\wn A,\A}
  \DP
\end{equation*}
This system is equivalent to \LL.

\subsection{Elementary Linear Logic}

Elementary Linear Logic (\ELL)~\cite{lll,djell} is a variant of \LL{} which has interesting computational complexity properties, since its cut elimination is shown to correspond to the elementary time complexity class (functions whose computation time is bounded by a tower of exponentials).
\ELL{} is obtained from \LL{} with functorial promotion by removing the ($\llde$) and ($\lldg$) rules:
\begin{equation*}
  \AIC{\vdash A,\A}
  \RL{\ocf}
  \UIC{\vdash \oc A,\wn\A}
  \DP
\qquad\qquad
  \AIC{\vdash\A}
  \RL{\llwk}
  \UIC{\vdash\wn A,\A}
  \DP
\qquad\qquad
  \AIC{\vdash\wn A,\wn A,\A}
  \RL{\llco}
  \UIC{\vdash\wn A,\A}
  \DP
\end{equation*}

\subsection{Soft Linear Logic}

Soft Linear Logic (\SLL)~\cite{sll} is obtained from \ELL{} by replacing the $\wn$-rules ($\llwk$) and ($\llco$) by a new family of rules called \emph{multiplexing rules} (for all $k\in\Nat$):
\begin{equation*}
  \AIC{\vdash A,\A}
  \RL{\ocf}
  \UIC{\vdash \oc A,\wn\A}
  \DP
\qquad\qquad
  \AIC{\vdash\overbrace{A, \dots, A}^k,\A}
  \RL{\mpx{k}}
  \UIC{\vdash\wn A,\A}
  \DP
\end{equation*}
The cases $k=0$ and $k=1$ give back ($\llwk$) and ($\llde$) of \LL, but for $k\geq 2$, we get different rules (in particular $\mpx{2}$ is not $\llco$).

The cut elimination of \SLL{} is related with the \PTIME{} complexity class~\cite{sll}.

\subsection{Light Linear Logic}

Light Linear Logic (\LLL)~\cite{lll} considers two different exponential signatures $\{\sgte,\epg\}$.
We use the notations $\oc A:=\oc_\sgte A$, $\wn A:=\wn_\sgte A$, $\pg A:=\oc_\epg A$ and $\copg A:=\wn_\epg A$.
The exponential rules are:
\begin{equation*}
  \AIC{\vdash A,B}
  \RL{\ocu}
  \UIC{\vdash\oc A,\wn B}
  \DP
\qquad\qquad
  \AIC{\vdash A,\A,\B}
  \RL{\pg}
  \UIC{\vdash \pg A,\copg\A,\wn\B}
  \DP
\qquad\qquad
  \AIC{\vdash\A}
  \RL{\llwk}
  \UIC{\vdash\wn A,\A}
  \DP
\qquad\qquad
  \AIC{\vdash\wn A,\wn A,\A}
  \RL{\llco}
  \UIC{\vdash\wn A,\A}
  \DP
\end{equation*}
We then have two kinds of promotions: unary functorial promotion ($\ocu$) for $\oc$, and $\pg$-promotion for $\pg$.

This system is also related with \PTIME{} complexity~\cite{lll}.

\subsection{Shifting Operators}

Shifting operators are a linear version of \LL's exponential modalities~\cite{locus}.
The system we consider here is also based on $\{\sgte,\epg\}$-formulas, but the standard notations are: $\oc A:=\oc_\sgte A$, $\wn A:=\wn_\sgte A$, $\shpos A:=\oc_\epg A$ and $\shneg A:=\wn_\epg A$.
The exponential rules extend those of \LL:
\begin{gather*}
  \AIC{\vdash A,\wn\A}
  \RL{\oc}
  \UIC{\vdash\oc A,\wn\A}
  \DP
\qquad\qquad
  \AIC{\vdash A,\A}
  \RL{\llde}
  \UIC{\vdash\wn A,\A}
  \DP
\qquad\qquad
  \AIC{\vdash\A}
  \RL{\llwk}
  \UIC{\vdash\wn A,\A}
  \DP
\qquad\qquad
  \AIC{\vdash\wn A,\wn A,\A}
  \RL{\llco}
  \UIC{\vdash\wn A,\A}
  \DP
\\[2ex]
  \AIC{\vdash A,\shneg\A}
  \RL{\shpos}
  \UIC{\vdash\shpos A,\shneg\A}
  \DP
\qquad\qquad
  \AIC{\vdash A,\A}
  \RL{\shneg}
  \UIC{\vdash\shneg A,\A}
  \DP
\end{gather*}

\subsection{Subexponentials}\label{secsubexp}

Subexponential Linear Logic (\seLL) denotes a family of systems which deal with multiple exponential signatures.
$\seLL(\sset,{\leqse},\sset_W,\sset_C)$~\cite{structexp,subexp} is a system with parameters:
\begin{itemize}
	\item $(\sset,{\leqse})$ is a pre-ordered set of \emph{exponential signatures}. So that formulas of $\seLL(\sset,{\leqse},\sset_W,\sset_C)$ are \sset-formulas and $\leqse$ plays a key role in the promotion rule.
	\item $\sset_W$ and $\sset_C$ are two subsets of $\sset$ used to control $\wn$-rules.
\end{itemize}
The exponential rules are:
\begin{gather*}
  \AIC{\vdash A,\wn_{e_1}B_1,\dotsc,\wn_{\e_n}B_n}
  \AIC{\e\leqse\e_1\quad\dotsb\quad\e\leqse e_n}
  \RL{\oc_\e}
  \BIC{\vdash \oc_\e A,\wn_{\e_1}B_1,\dotsc,\wn_{\e_n}B_n}
  \DP
  \\[2ex]
  \AIC{\vdash A,\A}
  \RL{\llde[\e]}
  \UIC{\vdash\wn_\e A,\A}
  \DP
  \qquad\qquad
  \AIC{\vdash\A}
  \AIC{\e\in\sset_W}
  \RL{\llwk[e]}
  \BIC{\vdash\wn_\e A,\A}
  \DP
  \qquad\qquad
  \AIC{\vdash\wn_\e A,\wn_\e A,\A}
  \AIC{\e\in\sset_C}
  \RL{\llco[\e]}
  \BIC{\vdash\wn_\e A,\A}
  \DP
\end{gather*}
For cut elimination to hold, some properties of the parameters must be requested:
\begin{thm}[Cut Elimination~\cite{structexp}]
If $\sset_W$ and $\sset_C$ are upward closed (\ie{} $\e\in\sset_W \Rightarrow \e\leqse\e' \Rightarrow \e'\in\sset_W$, and the same with $\sset_C$), then cut elimination holds.
\end{thm}

As a variant, the subexponential system presented in~\cite{undecidsubexp} is a particular case of the system above in which $\sset_W=\sset_C$.

\begin{rem}
  The instance of \seLL{} where $\sset=\{\sgte\}$ is a singleton, $\sgte\leqse\sgte$, and $\sset_W=\sset_C=\sset$ is \LL.

The instance of \seLL{} where $\sset=\{\sgte,\epg\}$, $\sgte\leqse\sgte$, $\epg\leqse\epg$, and $\sset_W=\sset_C=\{\sgte\}$ is \LL{} with shifting operators.
\end{rem}

\subsection{Stratified Bounded Linear Logic}

While \BsLL{} is presented in~\cite{llcoeffects} as an intuitionistic system, we consider here its (one-sided) classical version.
Everything we discuss in this paper could be done in an intuitionistic setting in a very similar way.

As in \seLL, \BsLL{} considers multiple exponential connectives.
In \BsLL, exponential signatures come with a richer algebraic structure.
\BsLL{} is parameterized by an ordered semi-ring $(\sset,{+},0,{\cdot},1,{\leqse})$.
Formulas are \sset-formulas, and the exponential rules are:
\begin{gather*}
  \AIC{\vdash A,\wn_{e_1}B_1,\dotsc,\wn_{\e_n}B_n}
  \RL{\oc_{\_\cdot\_}}
  \UIC{\vdash \oc_\e A,\wn_{\e\cdot\e_1}B_1,\dotsc,\wn_{\e\cdot\e_n}B_n}
  \DP
  \qquad\qquad
  \AIC{\vdash\wn_{\e_1}A,\A}
  \AIC{\e_1\leqse\e_2}
  \RL{{\leqse}}
  \BIC{\vdash\wn_{\e_2}A,\A}
  \DP
  \\[2ex]
  \AIC{\vdash A,\A}
  \RL{\llde[1]}
  \UIC{\vdash\wn_1A,\A}
  \DP
  \qquad\qquad
  \AIC{\vdash\A}
  \RL{\llwk[0]}
  \UIC{\vdash\wn_0A,\A}
  \DP
  \qquad\qquad
  \AIC{\vdash\wn_{\e_1}A,\wn_{\e_2}A,\A}
  \RL{\llco[{\_+\_}]}
  \UIC{\vdash\wn_{\e_1+\e_2}A,\A}
  \DP
\end{gather*}

\section{Super Linear Logic}\label{secsuperll}

We follow the ideas of subexponentials and bounded linear logic with parameters which try to subsume both.
Given a set \sset{} (the set of exponential signatures), we consider the following family of predicates:
\begin{equation*}
  \begin{array}{|@{\quad}c@{\quad}|@{\quad}c@{\quad}|@{\quad}c@{\quad}|@{\quad}c@{\quad}|}
    \hline
    & & & \\[-1.5ex]
    \depred:\sset\rightarrow\Bool & \copred_k:\sset^{k+1}\rightarrow\Bool\quad (\forall k\geq 0) & \dgpred:\sset^3\rightarrow\Bool & \prompred _n:\sset\rightarrow\Bool\quad (\forall n\geq 0)\\[1ex]
    \hline
  \end{array}
\end{equation*}

\begin{nota}
Given a predicate $\varphi:\sset^p\rightarrow\Bool$, we often write $\varphi(\e_1,\dotsc,\e_p)$ for $\varphi(\e_1,\dotsc,\e_p)=\true$.
\end{nota}

The system \superLL(\sset,\depred,\copred,\dgpred,\prompred) is defined by:
formulas are \sset-formulas, and the exponential rules are:
\begin{gather*}
  \AIC{\vdash A,\A}
  \AIC{\depred(\e)}
  \RL{\de}
  \BIC{\vdash\wn_\e A,\A}
  \DP
\qquad\qquad
  \AIC{\vdash\wn_{\e_1}A,\dotsc,\wn_{\e_k}A,\A}
  \AIC{\copred_k(\e_1,\dotsc,\e_k,\e)}
  \RL{\co{}}
  \BIC{\vdash\wn_\e A,\A}
  \DP
\\[2ex]
  \AIC{\vdash\wn_{\e_1}\wn_{\e_2}A,\A}
  \AIC{\dgpred(\e_1,\e_2,\e)}
  \RL{\dg}
  \BIC{\vdash\wn_\e A,\A}
  \DP
\qquad\qquad
  \AIC{\vdash A,A_1,\dotsc,A_n}
  \AIC{\prompred_n(\e)}
  \RL{\prom}
  \BIC{\vdash\oc_\e A,\wn_\e A_1,\dotsc,\wn_\e A_n}
  \DP
\end{gather*}

\begin{ex}
Let us detail the meaning of the ($\copred$) rule for $k=2$:
\begin{prooftree}
\AIC{\vdash \wn_{\e_1}A,\wn_{\e_2}A,\A}
\AIC{\copred_2(\e_1,\e_2,\e)}
\RL{\copred}
\BIC{\vdash\wn_\e A,\A}
\end{prooftree}
It tells us that: if ${}\vdash\wn_{\e_1}A,\wn_{\e_2}A,\A$ is derivable and $\e_1,\e_2,\e\in\sset$ are exponential signatures such that $\copred_2(\e_1,\e_2,\e)=\true$ then the rule applies and one can deduce ${}\vdash\wn_\e A,\A$.
It generalizes the usual contraction rule of \LL{} to a given relation $\copred_2$ relating the involved exponential signatures.
\end{ex}

Note that the weakening rule is incorporated in the ($\copred$) rule for $k=0$:
\begin{prooftree}
 \AIC{\vdash\A}
\AIC{\copred_0(\e)}
\RL{\copred}
\BIC{\vdash\wn_\e A,\A} 
\end{prooftree}

In the case $k=1$, the ($\copred$) rule acts as a \emph{subsumption rule}:
\begin{prooftree}
 \AIC{\vdash\wn_{\e_1}A,\A}
\AIC{\copred_1(\e_1,\e_2)}
\RL{\copred}
\BIC{\vdash\wn_{\e_2}A,\A} 
\end{prooftree}
with respect to the relation $\wn_{\e_1}A\leq\wn_{\e_2}A := \copred_1(\e_1,\e_2)$.
If $\copred_1$ is a subdiagonal relation (\ie{} $\copred_1(\e_1,\e_2)\Rightarrow\e_1=\e_2$), the ($\co{}$) rule for $k=1$ is trivial and can be omitted (in particular if \sset{} is a singleton).

($\prom$) corresponds to a functorial version of the promotion rule. $\prompred_n$ controls the width of the rule.

\begin{rem}
  \superLL{} should be considered as a refinement of \LL{} rather than an extension.
Indeed the forgetful function which maps formulas $\oc_\e A$ (resp.\ $\wn_\e A$) to $\oc A$ (resp.\ $\wn A$), maps any proof in \superLL{} into a proof of the corresponding sequent in \LL, since the induced rules are all derivable in \LL.
\end{rem}

\paragraph{Functional Instances.}
In the particular case where all the parameter relations $\depred$, $\copred_k$ ($k\neq 1$) and $\dgpred$ have their last element uniquely defined from the previous ones:
\begin{equation*}
\textsc{r}(\e_1,\dotsc,\e_n,\e) \rightarrow \textsc{r}(\e_1,\dotsc,\e_n,\e') \rightarrow \e = \e'
\end{equation*}
the instance is called \emph{functional}.

In particular there is at most one $\e$ such that $\depred(\e)$ in a functional instance.
We note it $1$ if it exists.
In the same spirit we use the notations $\_\times\_$ for the partial function induced by $\dgpred$ (\ie{} $\dgpred(\e_1,\e_2,\e_1\times\e_2)=\true$ if such an $\e_1\times\e_2$ exists), and $\_+_k\dotsb +_k\_$ for the partial function induced by $\copred_k$ (\ie{} $\copred_k(\e_1,\dotsc,\e_k,\e_1+_k\dotsb+_k\e_k)=\true$ if such an $\e_1+_k\dotsb+_k\e_k$ exists) for $k > 1$. The unique element $\e$ (if it exists) such that $\copred_0(\e)$ is noted $0$.

If $\sset$ is a singleton, the instance is immediately functional.

\section{Cut Elimination}\label{secce}

Let us now move to the key result we want to prove about \superLL: cut elimination.
As defined above, the system \superLL{} is not really meaningful.
Properties relating the parameters must be ensured to get a significant system, in particular regarding cut elimination.

\begin{ex}
Let us consider the instance $\sset=\{\e,\e'\}$, $\prompred_2(\e)=\true$, $\prompred_1(\e')=\true$ and $\copred_1(\e',\e)=\true$, but $\prompred_2(\e')=\false$.

We have the following derivation:
\begin{prooftree}
  \RL{\ax}
  \ZIC{\vdash X\orth,X}
  \ZIC{\prompred_1(\e')}
  \RL{\prom}
  \BIC{\vdash\oc_{\e'} X\orth, \wn_{\e'} X}
  \ZIC{\copred_1(\e',\e)}
  \RL{\copred}
  \BIC{\vdash\wn_\e X,\oc_{\e'} X\orth}
  \RL{\ax}
  \ZIC{\vdash X,X\orth}
  \RL{\ax}
  \ZIC{\vdash X,X\orth}
  \RL{\tens}
  \BIC{\vdash X\tens X, X\orth, X\orth}
  \ZIC{\prompred_2(\e)}
  \RL{\prom}
  \BIC{\vdash \oc_\e X\orth, \wn_\e (X\tens X), \wn_\e X\orth}
  \RL{\cut}
  \BIC{\vdash\oc_{\e'} X\orth, \wn_\e (X\tens X), \wn_\e X\orth}
\end{prooftree}
However it is not possible to find a cut-free proof of $\vdash\oc_{\e'} X\orth, \wn_\e (X\tens X), \wn_\e X\orth$.
\end{ex}

In order to explain the constraints we will put on the parameters defining \superLL, let us first give a sketch of the proof we are going to use for cut elimination.

\subsection{Proof Sketch}\label{seccesketch}

\begin{thmst}[Cut Elimination]
The ($\cut$) rule is admissible in the system without the ($\cut$) rule.
\end{thmst}

The global pattern of the proof we are going to use is folklore and it is the one used in the \yalla{} library~\cite{yalla}.
We prove that the ($\cut$) rule:
\begin{prooftree}
  \AIC{\pi_1}
  \noLine
  \UIC{\vdash A,\A}
  \AIC{\pi_2}
  \noLine
  \UIC{\vdash A\orth,\B}
  \RL{\cut}
  \BIC{\vdash\A,\B}
\end{prooftree}
is admissible by induction on the lexicographically ordered pair (size of $A$, size of $\pi_1$ + size of $\pi_2$):
\begin{itemize}
\item If $\pi_1$ or $\pi_2$ does not end with a rule acting on $A$, we apply the induction hypothesis with the premise(s) of this rule.
\item If both $\pi_1$ and $\pi_2$ end with non-exponential rules introducing the main connective of $A$ and $A\orth$, we can apply the induction hypothesis with smaller cut formulas. A typical example is:
\begin{equation*}
\hspace{-12pt}
\AIC{\pi'_1}
\noLine
\UIC{\vdash A,\A}
\AIC{\pi'_2}
\noLine
\UIC{\vdash B,\B}
\RL{\tens}
\BIC{\vdash A\tens B,\A,\B}
\AIC{\pi'_3}
\noLine
\UIC{\vdash A\orth,B\orth,\C}
\RL{\parr}
\UIC{\vdash A\orth\parr B\orth,\C}
\RL{\cut}
\BIC{\vdash\A,\B,\C}
\DP
\qquad\reduc\qquad
\AIC{\pi'_2}
\noLine
\UIC{\vdash B,\B}
\AIC{\pi'_1}
\noLine
\UIC{\vdash A,\A}
\AIC{\pi'_3}
\noLine
\UIC{\vdash A\orth,B\orth,\C}
\dashedLine
\RL{IH(A)}
\BIC{\vdash \A,B\orth,\C}
\dashedLine
\RL{IH(B)}
\BIC{\vdash \A,\B,\C}
\DP
\end{equation*}
\item If $\pi_1$ and $\pi_2$ both end with promotion rules, we have to deal with situations like:
\begin{prooftree}
  \AIC{\pi'_1}\noLine
  \UIC{\vdash A,B_1,B_2}
  \AIC{\prompred_2(\e)}
  \RL{\prompred}
  \BIC{\vdash\oc_\e A,\wn_\e B_1,\wn_\e B_2}
  \AIC{\pi'_2}\noLine
  \UIC{\vdash C,A\orth,D_1,D_2}
  \AIC{\prompred_3(\e)}
  \RL{\prom}
  \BIC{\vdash\oc_\e C,\wn_\e A\orth,\wn_\e D_1,\wn_\e D_2}
  \RL{\cut}
  \BIC{\vdash\oc_\e C,\wn_\e B_1,\wn_\e B_2,\wn_\e D_1,\wn_\e D_2}
\end{prooftree}
for which the most natural way to eliminate the cut is to build:
\begin{prooftree}
  \AIC{\pi'_1}\noLine
  \UIC{\vdash A,B_1,B_2}
  \AIC{\pi'_2}\noLine
  \UIC{\vdash C,A\orth,D_1,D_2}
\dashedLine
  \RL{IH(A)}
  \BIC{\vdash C,B_1,B_2,D_1,D_2}
  \AIC{\prompred_4(\e)}
  \RL{\prom}
  \BIC{\vdash\oc_\e C,\wn_\e B_1,\wn_\e B_2,\wn_\e D_1,\wn_\e D_2}
\end{prooftree}
but it then requires to be able to derive $\prompred_4(\e)$ (from $\prompred_2(\e)$ and $\prompred_3(\e)$).
This is one of the reasons for the axioms of Section~\ref{secceax}.
\item If $\pi_1$ ends with a promotion rule and $\pi_2$ ends with a ($\co{}$) rule acting on $A$, we have to deal with situations like:
\begin{prooftree}
  \AIC{\pi'_1}\noLine
  \UIC{\vdash A,B}
  \AIC{\prompred_1(\e)}
  \RL{\prom}
  \BIC{\vdash\oc_\e A,\wn_\e B}
  \AIC{\pi'_2}\noLine
  \UIC{\vdash\wn_{\e'}A\orth,\A}
  \AIC{\copred_1(\e',\e)}
  \RL{\co{}}
  \BIC{\vdash\wn_\e A\orth,\A}
  \RL{\cut}
  \BIC{\vdash\wn_\e B,\A}
\end{prooftree}
for which the most natural way to eliminate the cut is to build:
\begin{prooftree}
  \AIC{\pi'_1}\noLine
  \UIC{\vdash A,B}
  \AIC{\prompred_1(\e')}
  \RL{\prom}
  \BIC{\vdash\oc_{\e'}A,\wn_{\e'}B}
  \AIC{\pi'_2}\noLine
  \UIC{\vdash\wn_{\e'}A\orth,\A}
\dashedLine
  \RL{IH(\oc_{\e'}A)}
  \BIC{\vdash\wn_{\e'}B,\A}
  \AIC{\copred_1(\e',\e)}
  \RL{\co{}}
  \BIC{\vdash\wn_\e B,\A}
\end{prooftree}
but it then requires to be able to derive $\prompred_1(\e')$ (from $\prompred_1(\e)$ and $\copred_1(\e',\e)$).
This is one of the reasons for the axioms of Section~\ref{secceax}.
\item Other situations are more problematic:
\begin{prooftree}
  \AIC{\pi'_1}\noLine
  \UIC{\vdash A,B}
  \AIC{\prompred_1(\e)}
  \RL{\prom}
  \BIC{\vdash\oc_\e A,\wn_\e B}
  \AIC{\pi'_2}\noLine
  \UIC{\vdash\wn_{\e_1}\wn_{\e_2}A\orth,\A}
  \AIC{\dgpred(\e_1,\e_2,\e)}
  \RL{\dg}
  \BIC{\vdash\wn_\e A\orth,\A}
  \RL{\cut}
  \BIC{\vdash\wn_\e B,\A}
\end{prooftree}
for which the most natural way to eliminate the cut is to build:
\begin{prooftree}
  \AIC{\pi'_1}\noLine
  \UIC{\vdash A,B}
  \AIC{\prompred_1(\e_2)}
  \RL{\prom}
  \BIC{\vdash\oc_{\e_2}A,\wn_{\e_2}B}
  \AIC{\prompred_1(\e_1)}
  \RL{\prom}
  \BIC{\vdash\oc_{\e_1}\oc_{\e_2}A,\wn_{\e_1}\wn_{\e_2}B}
  \AIC{\pi'_2}\noLine
  \UIC{\vdash\wn_{\e_1}\wn_{\e_2}A\orth,\A}
\dashedLine
  \RL{IH(\oc_{\e_1}\oc_{\e_2}A)}
  \BIC{\vdash\wn_{\e_1}\wn_{\e_2}B,\A}
  \AIC{\dgpred(\e_1,\e_2,\e)}
  \RL{\dg}
  \BIC{\vdash\wn_\e B,\A}
\end{prooftree}
but the size of $\oc_{\e_1}\oc_{\e_2}A$ being bigger than the size of $\oc_{\e_1}A$ there is no valid way of applying the induction hypothesis.
This is why we need to use more global transformations of proofs when reducing cuts on exponential formulas. This is the purpose of the substitution lemma of Section~\ref{secsubstlem}.
\end{itemize}

\subsection{Cut-Elimination Axioms}\label{secceax}

The \emph{cut-elimination axioms} are the $3$ properties of the parameters \sset, \prompred, \copred{} and \dgpred{} presented in Table~\ref{tabceax}.

\renewcommand\theequation{ce\arabic{equation}}
\begin{table}
\begin{align}
  \forall m,n\in\Nat,\forall\e\in\sset,\quad m>0 \rightarrow \prompred_m(\e)\rightarrow{} &\prompred_n(\e)\rightarrow\prompred_{m+n-1}(\e) \label{axceprom}\\
  \forall k,n\in\Nat,\forall\e_1,\dotsc,\e_k,\e\in\sset,\quad \copred_k(\e_1,\dotsc,\e_k,\e)\rightarrow{} & \prompred_n(\e)\rightarrow\prompred_n(\e_1)\wedge\dotsb\wedge\prompred_n(\e_k) \label{axceco}\\
  \forall n\in\Nat,\forall\e_1,\e_2,\e\in\sset,\quad \dgpred(\e_1,\e_2,\e)\rightarrow{} & \prompred_n(\e)\rightarrow\prompred_n(\e_1)\wedge\prompred_n(\e_2) \label{axcedg}
\end{align}
\caption{Cut-Elimination Axioms}\label{tabceax}
\end{table}

Here are some important remarks about these axioms:
\begin{itemize}
\item For each $\e\in\sset$, axiom~(\ref{axceprom}) gives a closure property of the set $\{n\in\Nat \mid \prompred_n(\e)\}$ of natural numbers.
If $2$ belongs to this set, then it must be upward closed.
If $0$ belongs to this set, then it must be downward closed.
The full set $\Nat$ satisfies the axiom~(\ref{axceprom}), as well as $\{1\}$.
\item In axiom~(\ref{axceco}), the case $k=0$ is always valid.
\item For each $n\in\Nat$, axioms (\ref{axceco}) and (\ref{axcedg}) give closure properties for the set $\{\e\in\sset \mid \prompred_n(\e)\}$.
\item If $\sset$ is a singleton, axioms (\ref{axceco}) and (\ref{axcedg}) are satisfied.
\item If the relations $(\prompred_n)_{n\in\Nat}$ are full (\ie{} always true) then all the axioms of Table~\ref{tabceax} hold.
\end{itemize}


\subsection{Substitution Lemma}\label{secsubstlem}

In this section, we suppose that the parameters of $\superLL$ satisfy the cut-elimination axioms of Table~\ref{tabceax}.

As explained in Section~\ref{seccesketch}, using small step transformations does not allow us to apply our induction hypothesis for exponential cuts in the cut-elimination proof.
For this reason, we have to define a bigger step called \emph{substitution lemma}. It describes how to hereditary reduce the residuals of an exponential cut until the size of the cut formula strictly decreases.

\begin{nota}
If $\vec{\e}=\e_1,\dotsc,\e_n$, and $\textsc{r}$ is a predicate, then $\textsc{r}(\vec{\e})$ means that $\textsc{r}(\e_i)$ is true for all $1\leq i\leq n$.
\end{nota}

\begin{lem}[Substitution Lemma]
\label{subs}
Let $A$ be a formula, let $\B$ be a context, and let $\vec{\e^1},\dotsc,\vec{\e^s}$ be non-empty lists of signatures such that $\prompred_{\size{\B}}(\vec{\e^j})$ is \true{} for all $1\leq j\leq s$, and such that for all $\A$, if $\vdash A,\A$ is provable without using any cut then $\vdash\B,\A$ is provable without using any cut.
Then we have that for all $\A$, if $\vdash\wn_{\vec{\e^1}}A,\dotsc,\wn_{\vec{\e^s}}A,\A$ is provable without using any cut then $\vdash\wn_{\vec{\e^1}}\B,\dotsc,\wn_{\vec{\e^s}}\B,\A$ as well.
\end{lem}

\begin{proof}
First we can notice that for any $ \A $ the following rule:
\begin{gather*}
\AIC{\vdash A,\dotsc,A,\A}
\dashedLine
\RL{S}
\UIC{\vdash\B,\dotsc,\B,\A}
\DP
\end{gather*}
is admissible in the system without cuts (by using an easy induction on the number of $A$).

We can also notice that, for all $i\leq k$, we have:
\begin{gather*}
\AIC{\prompred_k(\e)}
\AIC{\prompred_m(\e)}
\dashedLine
\RL{M_i}
\BIC{\prompred_{k+(m-1)i}(\e)}
\DP
\end{gather*}
by simple induction on $i$ using axiom~(\ref{axceprom}).

Now we show the lemma by induction on the proof of $\vdash\wn_{\vec{\e^1}}A,\dots,\wn_{\vec{\e^s}}A,\A$. We distinguish cases according to the last applied rule:
\begin{itemize}
\item If it is a rule on a formula of $\A$ which is not a promotion:
\begin{equation*}
	\AIC{\pi}
	\noLine
	\RL{}
	\UIC{\vdash\wn_{\vec{\e^1}}A,\dotsc,\wn_{\vec{\e^s}}A,\A'}
	\RL{r}
	\UIC{\vdash\wn_{\vec{\e^1}}A,\dotsc,\wn_{\vec{\e^s}}A,\A}
	\DP
	\qquad\transfo\qquad
	\AIC{IH(\pi)}
	\noLine
	\RL{}
	\UIC{\vdash\wn_{\vec{\e^1}}\B,\dotsc,\wn_{\vec{\e^s}}\B,\A'}
	\RL{r}
	\UIC{\vdash\wn_{\vec{\e^1}}\B,\dotsc,\wn_{\vec{\e^s}}\B,\A}
	\DP
\end{equation*}
\item If it is a promotion introducing $\e$, all $\vec{\e^j}$ ($1\leq j\leq s$) start with $\e$.
Among them, we distinguish those of length $1$ (which are then restricted to $\e$): we assume $\vec{\e^j}=\e,\vec{\eb^j}$ ($1\leq j\leq p$) has at least two elements, and $\vec{e^{p+1}},\dotsc,\vec{e^s}$ are singletons:
\begin{multline*}
	\AIC{\pi}
	\noLine
	\UIC{\vdash B, \A', \wn_{\vec{\eb^1}}A, \dotsc, \wn_{\vec{\eb^p}}A, \overbrace{A, \dotsc, A}^{s-p}}
	\AIC{\prompred_{s+\size{\A'}}(\e)}
	\RL{\prom}
	\BIC{\vdash\oc_\e B, \wn_\e\A', \wn_{\vec{\e^1}}A, \dotsc, \wn_{\vec{\e^s}}A}
	\DP
	\qquad\transfo\qquad\\
	\AIC{IH(\pi)}
	\noLine
	\UIC{\vdash B, \A', \wn_{\vec{\eb^1}}\B, \dotsc, \wn_{\vec{\eb^p}}\B, \overbrace{A, \dotsc, A}^{s-p}}
	\dashedLine
	\RL{S}
	\UIC{\vdash B, \A', \wn_{\vec{\eb^1}}\B, \dotsc, \wn_{\vec{\eb^p}}\B, \overbrace{\B, \dotsc, \B}^{s-p}}
	\AIC{\prompred_{s+\size{\A'}}(\e)}
	\ZIC{\prompred_{\size{\B}}(\e)}
        \dashedLine
	\RL{M_s}
	\BIC{\prompred_{s\size{\B}+\size{\A'}}(\e)}
	\RL{\prom}
	\BIC{\vdash\oc_\e B,\wn_\e\A',\wn_{\vec{\e^1}}\B,\dotsc,\wn_{\vec{\e^s}}\B}
	\DP
\end{multline*}
\item If it is an ($\ax$) rule on $ \wn_{\vec{\e^1}}A$. Then $\A=\oc_{\vec{\e^1}}A\orth$ and we have:
  \begin{prooftree}
    \RL{\ax}
    \ZIC{\vdash A\orth, A}
    \RL{S}\dashedLine
    \UIC{\vdash A\orth, \B}
    \ZIC{\prompred_{\size{\B}}(\vec{e^1})}
    \doubleLine
    \RL{\prom}
    \BIC{\vdash \oc_{\vec{e^1}}A\orth, \wn_{\vec{e^1}} \B}
  \end{prooftree}
\item If it is a dereliction on $\wn_{\vec{\e^1}}A$, we have $\vec{\e^1}=\e,\vec{\eb}$:
	\begin{gather*}
	\hspace{-2cm}
		\AIC{\pi}
		\noLine
		\UIC{\vdash\wn_{\vec{\eb}}A,\wn_{\vec{\e^2}}A,\dotsc,\wn_{\vec{\e^s}}A,\A}
		\AIC{\depred(\e)}
		\RL{\de}
		\BIC{\vdash\wn_{\vec{\e^1}}A,\dotsc,\wn_{\vec{\e^s}}A,\A}
		\DP
	\qquad\transfo\qquad
		\AIC{IH(\pi)}
		\noLine
		\UIC{\vdash\wn_{\vec{\eb}}\B,\wn_{\vec{\e^2}}\B,\dotsc,\wn_{\vec{\e^s}}\B,\A}
		\AIC{\depred(\e)}
		\doubleLine
		\RL{\de}
		\BIC{\vdash\wn_{\vec{\e^1}}\B,\dotsc,\wn_{\vec{\e^s}}\B,\A}
		\DP
	\end{gather*}
\item If it is a contraction on $ \wn_{\vec{\e^1}}A$, we have $\vec{\e^1}=\e,\vec{\eb}$:
  \begin{prooftree}
    \AIC{\pi}
    \noLine
    \UIC{\vdash\wn_{\e_1}\wn_{\vec{\eb}}A,\dotsc,\wn_{\e_k}\wn_{\vec{\eb}}A,\wn_{\vec{\e^2}}A,\dotsc,\wn_{\vec{\e^s}}A,\A}
    \AIC{\copred_k(\e_1,\dotsc,\e_k,\e)}
    \RL{\co{}}
    \BIC{\wn_{\vec{\e^1}}A,\dotsc,\wn_{\vec{\e^s}}A,\A}
  \end{prooftree}
By axiom~(\ref{axceco}), we have $\prompred_{\size{\B}}(\e_1),\dotsc,\prompred_{\size{\B}}(\e_k)$, thus we can apply the induction hypothesis:
  \begin{prooftree}
    \AIC{IH(\pi)}\noLine
    \UIC{\vdash\wn_{\e_1}\wn_{\vec{\eb}}\B,\dotsc,\wn_{\e_k}\wn_{\vec{\eb}}\B,\wn_{\vec{\e^2}}\B,\dotsc,\wn_{\vec{\e^s}}\B,\A}
    \AIC{\copred_k(\e_1,\dotsc,\e_k,\e)}
    \doubleLine
    \RL{\co{}}
    \BIC{\vdash\wn_{\vec{\e^1}}\B,\dotsc,\wn_{\vec{\e^s}}\B,\A}
  \end{prooftree}
\item If it is a digging on $\wn_{\vec{\e^1}}A $, we have $\vec{\e^1}=\e,\vec{\eb}$:
  \begin{prooftree}
    \AIC{\pi}\noLine
    \UIC{\vdash\wn_{\e'}\wn_{\e''}\wn_{\vec{\eb}}A,\wn_{\vec{\e^2}}A,\dotsc,\wn_{\vec{\e^s}}A,\A}
    \AIC{\dgpred(\e',\e'',\e)}
    \RL{\dg}
    \BIC{\wn_{\vec{\e^1}}A,\dotsc,\wn_{\vec{e^s}}A,\A}
  \end{prooftree}
By axiom~(\ref{axcedg}), we have $\prompred_{\size{\B}}(\e')$ and $\prompred_{\size{\B}}(\e'')$, thus we can apply the induction hypothesis:
  \begin{prooftree}
    \AIC{IH(\pi)}\noLine
    \UIC{\vdash \wn_{\e'}\wn_{\e''}\wn_{\vec{\eb}}\B,\wn_{\vec{\e^2}}\B,\dotsc,\wn_{\vec{\e^s}}\B, \A}
    \AIC{\dgpred(\e',\e'',\e)}
    \doubleLine
    \RL{\dg}
    \BIC{\vdash\wn_{\vec{\e^1}}\B,\dotsc,\wn_{\vec{\e^s}}\B,\A}
  \end{prooftree}
  \qedhere
\end{itemize}
\end{proof}

\subsection{Cut-Elimination Proof}

\begin{thm}[Cut Elimination]\label{thmce}
Cut elimination holds for $\superLL(\sset,\depred,\copred,\dgpred,\prompred)$ as long as the instance satisfies the cut-elimination axioms of Table~\ref{tabceax}.
\end{thm}

\begin{proof}
As introduced in Section~\ref{seccesketch}, we prove the result by induction on the couple $(t, s)$ with lexicographic order, where $t$ is the size of the cut formula and $s$ is the sum of the sizes of the premises of the cut. We distinguish cases depending on the last rules of the premises of the cut:
\begin{itemize}
	  \item If one of the premises does not end with a rule acting on the cut formula, we apply the induction hypothesis with the premise(s) of this rule.
          \item If both last rules act on the cut formula which does not start with an exponential connective, we apply the standard reduction steps for non-exponential cuts leading to cuts involving strictly smaller cut formulas. We conclude by applying the induction hypothesis.
          \item If we have an exponential cut for which the cut formula $\oc_\e A\orth$ is not the conclusion of a promotion rule introducing $\oc_\e$, the rule above $\oc_\e A\orth$ cannot be a promotion rule and we apply the induction hypothesis to its premise(s).
          \item If we have an exponential cut for which the cut formula $\oc_\e A\orth$ is the conclusion of a promotion rule. We can apply:
		\begin{equation*}
			\AIC{\vdash A\orth,\B}
			\AIC{\prompred_{\size{\B}}(\e)}
			\RL{\prom}
			\BIC{\vdash\oc_\e A\orth,\wn_\e\B}
			\AIC{\vdash\wn_\e A,\A}
			\RL{\cut}
			\BIC{\vdash\wn_\e\B,\A}
			\DP
			\qquad\reduc\qquad
			\AIC{\vdash\wn_\e A,\A}
			\AIC{\prompred_{\size{\B}}(\e)}
			\dashedLine
			\RL{\text{Lem.}~\ref{subs}}
			\BIC{\vdash\wn_\e\B,\A}
			\DP
		\end{equation*}
		We have that $A$ and $\B$ are such that for every $\A$ such that $\vdash A,\A$ is provable without cuts, $\vdash\B,\A$ too. Indeed, $A$ and $\B$ are such that $\vdash A\orth,\B$ is provable without cuts and we can apply the induction hypothesis (smaller cut formula). Therefore we can apply Lemma~\ref{subs} on $\vdash\wn_\e A,\A$ and obtain that $\vdash\wn_\e\B,\A$ is provable without cut.
\qedhere
\end{itemize}
\end{proof}

\section{Other Proof Transformations}\label{sectransfo}

\subsection{Axiom Expansion}

We consider now a much simpler property which is axiom expansion,
to show how it also provides natural constraints on the parameters of \superLL.

\begin{table}
\begin{equation}
  \forall\e\in\sset, \qquad \prompred_1(\e)\tag{ea}\label{axea}
\end{equation}
  \caption{Expansion Axiom}\label{tabexpax}
\end{table}

\begin{lem}[One-Step Axiom Expansion]\label{lemoneae}
If $\e$ is an exponential signature such that $\prompred_1(\e)=\true$, then the one-step axiom expansion holds for formulas $\wn_\e A$ and $\oc_\e A\orth$ in $\superLL$.
That is we can derive $\vdash\oc_\e A\orth,\wn_\e A$ from $\vdash A\orth,A$.
\end{lem}

\begin{proof}\leavevmode
\begin{equation*}
  \AIC{\vdash A,A\orth}
  \ZIC{\prompred_1(\e)}
  \RL{\prom}
  \BIC{\vdash\oc_\e A,\wn_\e A\orth}
  \DP
\qedhere
\end{equation*}
\end{proof}

\begin{prop}[Axiom Expansion]\label{propae}
  If $\sset$ satisfies the axiom (\ref{axea}) of Table~\ref{tabexpax}, then axiom expansion holds for $\superLL(\sset,\depred,\copred,\dgpred,\prompred)$, \ie{} $\vdash A,A\orth$ is derivable for any $A$ from the axiom rule restricted to $\vdash X, X\orth$.
\end{prop}


\subsection{Girardization}\label{secgir}

A key ingredient of Girard's original presentation of linear logic is the following promotion rule:
\begin{prooftree}
  \AIC{\vdash A,\wn\A}
  \RL{\oc}
  \UIC{\vdash\oc A,\wn\A}
\end{prooftree}
It leads to the sub-formula property while the digging rule immediately breaks it.
It is thus important to understand in which situations it is possible to replace the ``functorial promotion plus digging'' style used in \superLL{} by a Girard's style promotion rule.

Our approach is to find commutation axioms allowing to migrate digging rules up towards promotions in order to generate Girard's style promotion rules.
In the setting of \superLL, we call \emph{Girard's promotion} the following rule:
\begin{prooftree}
  \AIC{\vdash A,\wn_{\eb_1}A_1,\dotsc,\wn_{\eb_n}A_n}
  \AIC{\dgpred(\e,\eb_1,\e_1)\quad\dotsb\quad\dgpred(\e,\eb_n,\e_n)}
  \AIC{\prompred_n(\e)}
  \RL{\promg}
  \TIC{\vdash\oc_\e A,\wn_{\e_1}A_1,\dotsc,\wn_{\e_n}A_n}
\end{prooftree}
The commutation axioms we have to consider are the Girardization axioms presented in Table~\ref{tabgir}.

\stepcounter{myequation}
\renewcommand\theequation{gir\arabic{equation}}
\begin{table}
\begin{align}
  \forall\e_1,\e_2,\e\in\sset,\quad &\dgpred(\e_1,\e_2,\e)\rightarrow \prompred_1(\e_1) \label{axgirax} \\
  \forall\e_1,\e_2,\e\in\sset,\quad\depred(\e_1) \rightarrow{} & \dgpred(\e_1,\e_2,\e)\rightarrow \copred_1(\e_2,\e) \label{axgirde} \\
\begin{split}
  \forall k\in\Nat, \forall\eb_1,\dotsc,\eb_k,\e_1,\e_2,\e\in\sset,\quad \copred_k(\eb_1,\dotsc,\eb_k,\e_1) \rightarrow{} & \dgpred(\e_1,\e_2,\e)\rightarrow \\
&\hspace{-5cm}\exists\eb'_1,\dotsc,\eb'_k\in\sset,\dgpred(\eb_1,\e_2,\eb'_1)\wedge\dotsb\wedge\dgpred(\eb_k,\e_2,\eb'_k)\wedge\copred_k(\eb'_1,\dotsc,\eb'_k,\e)
\end{split} \label{axgirco} \\
  \forall\e_1,\e_2,\e_3,\e,\e'\in\sset,\qquad \dgpred(\e_1,\e_2,\e')\rightarrow\dgpred(\e',\e_3,\e)\rightarrow
&\exists\e''\in\sset, \dgpred(\e_2,\e_3,\e'')\wedge\dgpred(\e_1,\e'',\e) \label{axgirdg} \\
  \forall n\in\Nat, \forall\e\in\sset,\qquad n>0\rightarrow{} & \prompred_n(\e)\rightarrow \exists\e'\in\sset, \depred(\e')\wedge\dgpred(\e,\e',\e) \label{axgirdedg}
\end{align}
\caption{Girardization axioms}\label{tabgir}
\end{table}

\begin{rem}\label{remgirax}
  It is easier to get some intuition on the Girardization axioms if we consider a functional instance. In this particular case they are closed to properties of (partial) semi-rings.
  \begin{align*}
    \forall\e,\quad \copred_1(\e,&1\times\e) \tag{\ref{axgirde}} \\
    \forall\e,\quad 0\times\e &= 0 \tag{\ref{axgirco}} \\
    \forall k\geq 2,\forall\eb_1,\dotsc,\eb_k,\e,\quad (\eb_1+_k\dotsb+_k\eb_k)\times\e &= (\eb_1\times\e) +_k\dotsb+_k (\eb_k\times\e) \tag{\ref{axgirco}} \\
     \forall\e_1,\e_2,\e_3,\quad (\e_1\times\e_2)\times\e_3 &= \e_1\times(\e_2\times\e_3)\tag{\ref{axgirdg}} \\
  \forall n,\e,\quad n>0\rightarrow\prompred_n(\e)\rightarrow\quad \e\times 1 & = \e \tag{\ref{axgirdedg}}
  \end{align*}
  Moreover (\ref{axgirax}) is an immediate consequence of (\ref{axea}).
\end{rem}

\begin{lem}[Admissibility of Digging]\label{lemdgadmiss}
If we consider an instance of \superLL{} which satisfies the Girardization axioms (Table~\ref{tabgir}), and if moreover we replace the functorial promotion rule ($\prom$) by Girard's promotion rule ($\promg$) in the system, then the ($\dg$) rule is admissible in the obtained system.
\end{lem}

\begin{proof}
We prove that, given a proof $\pi$ with conclusion $\vdash\wn_{\e_1}\wn_\e A,\dotsc,\wn_{\e_n}\wn_\e A,\A$, if $\dgpred(\e_i,\e,\e'_i)=\true$ ($1\leq i\leq n$), then we can build a proof of $\vdash\wn_{\e'_1}A,\dotsc,\wn_{\e'_n}A,\A$
which uses neither functorial promotion nor digging.
This is done by induction on the size of $\pi$.
\begin{itemize} 
\item If the last rule of $\pi$ does not act on the $\wn_{\e_i}\wn_\e A$, we apply the induction hypothesis on the premises and we conclude.
\item If the last rule of $\pi$ is an ($\ax$) rule, we consider the following transformation:
\begin{equation*}
  \RL{\ax}
  \ZIC{\vdash \oc_{\e_1}\oc_\e A\orth, \wn_{\e_1}\wn_\e A}
  \DP
\qquad\transfo\qquad
  \RL{\ax}
  \ZIC{\vdash \oc_\e A\orth,\wn_\e A}
  \ZIC{\dgpred(\e_1,\e,\e'_1)}
  \ZIC{\dgpred(\e_1,\e,\e'_1)}
  \RL{\ref{axgirax}}
  \UIC{\prompred_1(\e_1)}
  \RL{\promg}
  \TIC{\vdash\oc_{\e_1}\oc_\e A\orth,\wn_{\e'_1}A}
  \DP
\end{equation*}
\item If the last rule of $\pi$ is a ($\de$) rule introducing $\wn_{\e_1}$ (it is similar for another $\wn_{\e_i}$):
\begin{prooftree}
  \AIC{\vdash\wn_\e A,\wn_{\e_2}\wn_\e A,\dotsc,\wn_{\e_n}\wn_\e A,\A}
  \AIC{\depred(\e_1)}
  \RL{\de}
  \BIC{\vdash\wn_{\e_1}\wn_\e A,\wn_{\e_2}\wn_\e A,\dotsc,\wn_{\e_n}\wn_\e A,\A}
\end{prooftree}
we use the induction hypothesis to build:
\begin{prooftree}
  \AIC{IH}\noLine
  \UIC{\vdash\wn_\e A,\wn_{\e'_2}A,\dotsc,\wn_{\e'_n}A,\A}
  \AIC{\depred(\e_1)}
  \ZIC{\dgpred(\e_1,\e,\e'_1)}
  \RL{\ref{axgirde}}
  \BIC{\copred_1(\e,\e'_1)}
  \RL{\co{}}
  \BIC{\vdash\wn_{\e'_1} A,\wn_{\e'_2}A,\dotsc,\wn_{\e'_n}A,\A}
\end{prooftree}
\item If the last rule of $\pi$ is a ($\co{}$) rule:
  \begin{prooftree}
  \AIC{\vdash\wn_{\eb_1}\wn_\e A,\dotsc,\wn_{\eb_k}\wn_\e A,\wn_{\e_2}\wn_\e A,\dotsc,\wn_{\e_n}\wn_\e A,\A}
  \AIC{\copred_k(\eb_1,\dotsc,\eb_k,\e_1)}
  \RL{\co{}}
  \BIC{\vdash\wn_{\e_1}\wn_\e A,\wn_{\e_2}\wn_\e A,\dotsc,\wn_{\e_n}\wn_\e A,\A}
  \end{prooftree}
by (\ref{axgirco}), we have $\dgpred(\eb_j,\e,\eb'_j)=\true$ ($1\leq j\leq k$), and we can use the induction hypothesis to build:
\begin{prooftree}
  \AIC{IH}\noLine
  \UIC{\vdash\wn_{\eb'_1}A,\dotsc,\wn_{\eb'_k}A,\wn_{\e'_2}A,\dotsc,\wn_{\e'_n}A,\A}
  \AIC{\copred_k(\eb_1,\dotsc,\eb_k,\e_1)}
  \ZIC{\dgpred(\e_1,\e,\e'_1)}
  \RL{\ref{axgirco}}
  \BIC{\copred_k(\eb'_1,\dotsc,\eb'_k,\e'_1)}
  \RL{\co{}}
  \BIC{\vdash\wn_{\e'_1}A,\wn_{\e'_2}A,\dotsc,\wn_{\e'_n}A,\A}
\end{prooftree}
\item If the last rule of $\pi$ is a Girard's style promotion:
\begin{prooftree}
  \AIC{\vdash C,\wn_{\eb_1}\wn_\e A,\dotsc,\wn_{\eb_n}\wn_\e A,\wn_{\et_1}B_1,\dotsc,\wn_{\et_m}B_m}
  \AIC{1\leq i\leq n}
  \noLine
  \UIC{\dgpred(\e',\eb_i,\e_i)}
  \AIC{1\leq j\leq m}
  \noLine
  \UIC{\dgpred(\e',\et_j,\et'_j)}
  \AIC{\prompred_{n+m}(\e')}
  \RL{\promg}
  \QIC{\vdash\oc_{\e'}C,\wn_{\e_1}\wn_\e A,\dotsc,\wn_{\e_n}\wn_\e A,\wn_{\et'_1}B_1,\dotsc,\wn_{\et'_m}B_m}
\end{prooftree}
by (\ref{axgirdg}), we have $\dgpred(\eb_i,\e,\eb'_i)=\true$ ($1\leq i\leq n$), and we can use the induction hypothesis to build:
\begin{equation*}
\hspace{-1.5cm}
  \AIC{IH}\noLine
  \UIC{\vdash C,\wn_{\eb'_1}A,\dotsc,\wn_{\eb'_n}A,\wn_{\et_1}B_1,\dotsc,\wn_{\et_m}B_m}
  \AIC{1\leq i\leq n}\noLine
  \UIC{\dgpred(\e',\eb_i,\e_i)}
  \ZIC{\dgpred(\e_i,\e,\e'_i)}
  \RL{\ref{axgirdg}}
  \BIC{\dgpred(\e',\eb'_i,\e'_i)}
  \AIC{1\leq j\leq m}
  \noLine
  \UIC{\dgpred(\e',\et_j,\et'_j)}
  \AIC{\prompred_{n+m}(\e')}
  \RL{\promg}
  \QIC{\vdash\oc_{\e'}C,\wn_{\e'_1}A,\dotsc,\wn_{\e'_n}A,\wn_{\et'_1}B_1,\dotsc,\wn_{\et'_m}B_m}
  \DP
\end{equation*}
\end{itemize}
The admissibility of ($\dg$) is then the particular case $n=1$.
\end{proof}

\begin{prop}[Girardization]\label{propgir}
  If an instance of \superLL{} satisfies the Girardization axioms (Table~\ref{tabgir}), then any proof can be replaced by a proof of the same sequent which uses neither the functorial promotion rule nor the digging rule, but Girard's promotion instead.
\end{prop}

\begin{proof}
The first step is to transform any functorial promotion rule into the associated Girard's promotion:
\begin{multline*}
  \AIC{\vdash A,A_1,\dotsc,A_n}
  \AIC{\prompred_n(\e)}
  \RL{\prom}
  \BIC{\vdash\oc_\e A,\wn_\e A_1,\dotsc,\wn_\e A_n}
  \DP
\qquad\transfo\qquad\\
  \AIC{\vdash A,A_1,\dotsc,A_n}
  \AIC{\prompred_n(\e)}
  \RL{\ref{axgirdedg}}
  \UIC{\depred(\e')}
  \doubleLine
  \RL{\de}
  \BIC{\vdash A,\wn_{\e'}A_1,\dotsc,\wn_{\e'}A_n}
  \AIC{\prompred_n(\e)}
  \RL{\ref{axgirdedg}}
  \UIC{\dgpred(\e,\e',\e)}
  \AIC{\prompred_n(\e)}
  \RL{\promg}
  \TIC{\vdash\oc_\e A,\wn_\e A_1,\dotsc,\wn_\e A_n}
  \DP
\end{multline*}
Then, we conclude by induction on the number of digging rules in the proof, by applying Lemma~\ref{lemdgadmiss}.
\end{proof}
It is important to notice that if the starting proof is cut-free then the obtained one as well.

\subsection{Subsumption Elimination}

We have already mentioned that, in the case $k=1$, the ($\co{}$) rule acts as a subsumption rule with respect to the binary relation $\copred_1(\_,\_)$.
Such a rule explicitly appears in \BsLL.
In \seLL, an order relation is involved as well but it is mostly attached to the promotion rule.
In our setting, such an ordered promotion rule is:
\begin{prooftree}
  \AIC{\vdash A,A_1,\dotsc,A_n}
  \AIC{\e\leq\e_1\quad\dotsb\quad\e\leq\e_n}
  \AIC{\prompred_n(\e)}
  \RL{\promleq}
  \TIC{\vdash\oc_\e A,\wn_{\e_1}A_1,\dotsc,\wn_{\e_n}A_n}
\end{prooftree}
where we use the notation $\e\leq\e'$ for $\copred_1(\e,\e')$ (and we will do so in all this section).

Under some hypotheses, it is possible to merge the subsumption rule (($\co{}$) with $k=1$) into the promotion rule to get the ordered promotion rule.
The required properties are presented in Table~\ref{tabsubsum}.

\stepcounter{myequation}
\renewcommand\theequation{sb\arabic{equation}}
\begin{table}
\begin{align}
  \forall\e\in\sset,\quad & \prompred_1(\e) \tag{\ref{axea}} \label{axsubax} \\
  \forall\e\in\sset,\qquad & \e\leq\e \label{axsubrefl}\\
  \forall\e_1,\e_2,\e_3\in\sset,\qquad & \e_1\leq\e_2\rightarrow\e_2\leq\e_3\rightarrow\e_1\leq\e_3 \label{axsubtrans}\\
  \forall\e_1,\e_2\in\sset,\quad\depred(\e_1) \rightarrow{} & \e_1\leq\e_2\rightarrow \depred(\e_2) \label{axsubde}\\
  \begin{split}
  \forall\eb_1,\dotsc,\eb_k,\e_1,\e_2\in\sset,\quad \copred_k(\eb_1,\dotsc,\eb_k,\e_1) \rightarrow{} & \e_1\leq\e_2\rightarrow \\
&\hspace{-3cm}\exists\eb'_1,\dotsc,\eb'_k\in\sset,\eb_1\leq\eb'_1\wedge\dotsb\wedge\eb_k\leq\eb'_k\wedge\copred_k(\eb'_1,\dotsc,\eb'_k,\e_2)
  \end{split} \label{axsubco}\\
  \forall\eb_1,\eb_2,\e_1,\e_2\in\sset,\quad  \dgpred(\eb_1,\eb_2,\e_1)\rightarrow{} &\e_1\leq\e_2 \rightarrow \exists\eb'_1\in\sset, \eb_1\leq\eb'_1\wedge\dgpred(\eb'_1,\eb_2,\e_2) \label{axsubdg}
\end{align}
\caption{Subsumption Axioms (with $\e\leq\e' := \copred_1(\e,\e')$)}\label{tabsubsum}
\end{table}

We can make a few comments about the axioms:
\begin{itemize}
  \item Axiom~(\ref{axsubrefl}) is reflexivity of $\copred_1$ and axiom~(\ref{axsubtrans}) is transitivity of $\copred_1$, so that $\copred_1$ has then to be a pre-order relation.
  \item Axiom~(\ref{axsubde}) is closure of $\depred$ under $\copred_1$.
  \item Axioms~(\ref{axsubco}) and~(\ref{axsubdg}) are commutation axioms. Axiom~(\ref{axsubco}) is trivial for $k=1$.
\end{itemize}

\begin{lem}[Admissibility of Subsumption]\label{lemsubadmiss}
If we consider an instance of \superLL{} which satisfies the subsumption axioms (Table~\ref{tabsubsum}), and if moreover we replace the functorial promotion rule ($\prom$) by the ordered promotion rule ($\promleq$) in the system, then the ($\co{}$) rule for $k=1$ is admissible in the obtained system.
\end{lem}

\begin{proof}
We prove that, given a proof $\pi$ with conclusion $\vdash\wn_{\e_1}A_1,\dotsc,\wn_{\e_n}A_n,\A$, if $\e_i\leq\e'_i$ ($1\leq i\leq n$), then we can build a proof of $\vdash\wn_{\e'_1}A_1,\dotsc,\wn_{\e'_n}A_n,\A$
which uses neither functorial promotion nor subsumption.
This is done by induction on the size of $\pi$.
\begin{itemize} 
\item If the last rule of $\pi$ does not act on the $\wn_{\e_i}A_i$, we apply the induction hypothesis on the premises and we conclude.
\item If the last rule of $\pi$ is an ($\ax$) rule, we consider the following transformation:
\begin{equation*}
  \RL{\ax}
  \ZIC{\vdash\oc_{\e_1}A_1\orth,\wn_{\e_1}A_1}
  \DP
\qquad\transfo\qquad
  \RL{\ax}
  \ZIC{\vdash A_1\orth,A_1}
  \ZIC{\e_1\leq\e'_1}
  \RL{\ref{axea}}
  \ZIC{\prompred_1(\e_1)}
  \RL{\promleq}
  \TIC{\vdash\oc_{\e_1}A_1\orth,\wn_{\e'_1}A_1}
  \DP
\end{equation*}
\item If the last rule of $\pi$ is a ($\de$) rule introducing $\wn_{\e_1}$ (it is similar for another $\wn_{\e_i}$):
\begin{prooftree}
  \AIC{\vdash A_1,\wn_{\e_2}A_2,\dotsc,\wn_{\e_n}A_n,\A}
  \AIC{\depred(\e_1)}
  \RL{\de}
  \BIC{\vdash\wn_{\e_1}A_1,\wn_{\e_2}A_2,\dotsc,\wn_{\e_n}A_n,\A}
\end{prooftree}
we use the induction hypothesis to build:
\begin{prooftree}
  \AIC{IH}\noLine
  \UIC{\vdash A_1,\wn_{\e'_2}A_2,\dotsc,\wn_{\e'_n}A_n,\A}
  \AIC{\depred(\e_1)}
  \ZIC{\e_1\leq\e'_1}
  \RL{\ref{axsubde}}
  \BIC{\depred(\e'_1)}
  \RL{\de}
  \BIC{\vdash\wn_{\e'_1}A_1,\wn_{\e'_2}A_2,\dotsc,\wn_{\e'_n}A_n,\A}
  \end{prooftree}
\item If the last rule of $\pi$ is a ($\co{}$) rule:
  \begin{prooftree}
  \AIC{\vdash\wn_{\eb_1}A_1,\dotsc,\wn_{\eb_k}A_1,\wn_{\e_2}A_2,\dotsc,\wn_{\e_n}A_n,\A}
  \AIC{\copred_k(\eb_1,\dotsc,\eb_k,\e_1)}
  \RL{\co{}}
  \BIC{\vdash\wn_{\e_1}A_1,\wn_{\e_2}A_2,\dotsc,\wn_{\e_n}A_n,\A}
  \end{prooftree}
by (\ref{axsubco}), we have $\eb_j\leq\eb'_j$ ($1\leq j\leq k$), and we can use the induction hypothesis to build:
  \begin{prooftree}
  \AIC{IH}\noLine
  \UIC{\vdash\wn_{\eb'_1}A_1,\dotsc,\wn_{\eb'_k}A_1,\wn_{\e'_2}A_2,\dotsc,\wn_{\e'_n}A_n,\A}
  \AIC{\copred_k(\eb_1,\dotsc,\eb_k,\e_1)}
  \ZIC{\e_1\leq\e'_1}
  \RL{\ref{axsubco}}
  \BIC{\copred_k(\eb'_1,\dotsc,\eb'_k,\e'_1)}
  \BIC{\vdash\wn_{\e'_1}A_1,\wn_{\e'_2}A_2,\dotsc,\wn_{\e'_n}A_n,\A}
  \end{prooftree}
\item If the last rule of $\pi$ is a ($\dg$) rule:
\begin{prooftree}
  \AIC{\vdash\wn_{\eb}\wn_{\eb'}A_1,\wn_{\e_2}A_2,\dotsc,\wn_{\e_n}A_n,\A}
  \AIC{\dgpred(\eb,\eb',\e_1)}
  \RL{\dg}
  \BIC{\vdash\wn_{\e_1}A_1,\wn_{\e_2}A_2,\dotsc,\wn_{\e_n}A_n,\A}
\end{prooftree}
by (\ref{axsubdg}), we have $\eb\leq\eb''$, and we can use the induction hypothesis to build:
\begin{prooftree}
  \AIC{IH}\noLine
  \UIC{\vdash\wn_{\eb''}\wn_{\eb'}A_1,\wn_{\e'_2}A_2,\dotsc,\wn_{\e'_n}A_n,\A}
  \AIC{\dgpred(\eb,\eb',\e_1)}
  \ZIC{\e_1\leq\e'_1}
  \RL{\ref{axsubdg}}
  \BIC{\dgpred(\eb'',\eb',\e'_1)}
  \RL{\dg}
  \BIC{\vdash\wn_{\e'_1}A_1,\wn_{\e'_2}A_2,\dotsc,\wn_{\e'_n}A_n,\A}
  \end{prooftree}
\item If the last rule of $\pi$ is an ordered promotion:
\begin{prooftree}
  \AIC{\vdash C,A_1,\dotsc,A_n,B_1,\dotsc,B_m}
  \AIC{1\leq i\leq n}
  \noLine
  \UIC{\e\leq\e_i}
  \AIC{1\leq j\leq m}
  \noLine
  \UIC{\e\leq\eb_j}
  \AIC{\prompred_{n+m}(\e)}
  \RL{\promleq}
  \QIC{\vdash\oc_\e C,\wn_{\e_1}A_1,\dotsc,\wn_{\e_n}A_n,\wn_{\eb_1}B_1,\dotsc,\wn_{\eb_m}B_m}
\end{prooftree}
we can build:
\begin{equation*}
\hspace{-1.5cm}
  \AIC{\vdash C,A_1,\dotsc,A_n,B_1,\dotsc,B_m}
  \AIC{1\leq i\leq n}\noLine
  \UIC{\e\leq\e_i}
  \ZIC{\e_i\leq\e'_i}
  \RL{\ref{axsubtrans}}
  \BIC{\e\leq\e'_i}
  \AIC{1\leq j\leq m}
  \noLine
  \UIC{\e\leq\eb_j}
  \AIC{\prompred_{n+m}(\e)}
  \RL{\promleq}
  \QIC{\vdash\oc_\e C,\wn_{\e'_1}A_1,\dotsc,\wn_{\e'_n}A_n,\wn_{\eb_1}B_1,\dotsc,\wn_{\eb_m}B_m}
  \DP\qedhere
\end{equation*}
\end{itemize}
\end{proof}

\begin{prop}[Subsumption Elimination]
\label{subselim}
  If an instance of \superLL{} satisfies the subsumption axioms (Table~\ref{tabsubsum}), then any proof can be replaced by a proof of the same sequent which uses neither the functorial promotion rule nor the subsumption rule, but the ordered promotion instead.
\end{prop}

\begin{proof}
The first step is to transform any functorial promotion rule into the associated ordered promotion:
\begin{equation*}
  \AIC{\vdash A,A_1,\dotsc,A_n}
  \AIC{\prompred_n(\e)}
  \RL{\prom}
  \BIC{\vdash\oc_\e A,\wn_\e A_1,\dotsc,\wn_\e A_n}
  \DP
\qquad\transfo\qquad
  \AIC{\vdash A,A_1,\dotsc,A_n}
  \RL{\ref{axsubrefl}}
  \ZIC{\e\leq\e}
  \AIC{\prompred_n(\e)}
  \RL{\promleq}
  \TIC{\vdash\oc_\e A,\wn_\e A_1,\dotsc,\wn_\e A_n}
  \DP
\end{equation*}
We conclude by induction on the number of subsumption rules in the proof, by applying Lemma~\ref{lemsubadmiss}.
\end{proof}

Again if the starting proof is cut-free then the obtained one as well.

\section{Sub-Systems}\label{secsubsyst}

Since \superLL{} depends on various parameters, it covers many possible systems through the choice of instances.

In the previous sections, we have seen some (mostly independent) sets of axioms which allow us to do proof manipulations leading to alternative rules for the system.
These proof transformations are the key tool to show how particular instances of \superLL{} are equivalent to known systems from the literature.

We now focus on specific choices of \sset, \depred, \copred, \dgpred{} and \prompred{} which give back known systems from Section~\ref{secmiscll}.
In each case we provide the appropriate values of the parameters to get the desired system. Moreover we check in each case that the cut-elimination axioms (Table~\ref{tabceax}) and the expansion axiom (\ref{axea}) are satisfied.

\begin{rem}\label{remco}
If $\e$ is an exponential signature, requiring $\copred_0(\e)$ and $\copred_2(\e, \e, \e)$ to be \true{}, or for all $k$ in $\Nat$, $\copred_k(\e,\dotsc,\e,\e)=\true$, leads to equivalent systems since the $k$-ary ($\co{}$) rule becomes derivable:
\begin{equation*}
\begin{array}{c@{\qquad\;}c@{\qquad\;}c}
  k=0 & k=1 & k\geq 2 \\[1ex]
  \AIC{\vdash\A}
  \AIC{\copred_0(\e)}
  \RL{\co{}}
  \BIC{\vdash\wn_\e A,\A}
  \DP
  &
  \AIC{\vdash\wn_{\e}A,\A}
  \DP
  &
  \AIC{\vdash\overbrace{\wn_{\e}A,\dotsc,\wn_{\e}A}^k,\A}
  \AIC{\copred_2(\e,\e,\e)}
  \RL{\co{}}
  \BIC{\vdash\overbrace{\wn_{\e}A,\dotsc,\wn_{\e}A}^{k-1},\A}
  \noLine
  \UIC{\vdots}
  \noLine
  \UIC{\vdash\wn_\e A,\wn_\e A,\A}
  \AIC{\copred_2(\e,\e,\e)}
  \RL{\co{}}
  \BIC{\vdash\wn_\e A,\A}
  \DP
\end{array}
\end{equation*}
\end{rem}

\subsection{\LL{} with Functorial Promotion}\label{secllfuncinst}

The definition of \superLL{} is based on a functorial version of the promotion rule.
It is thus not very surprising that the easiest system to find back inside \superLL{} is the ``functorial promotion + digging'' presentation of \LL.
We consider the instance given by (when describing instances, we list the values which make the predicates $\true$, all other combinations are $\false$):
\begin{equation*}
  \begin{array}{|c|c|c|c|c|}
    \hline
    \sset & \depred & \copred & \dgpred & \prompred \\
    \hline
    \{\sgte\} & \depred(\sgte) & \copred_0(\sgte)\qquad \copred_2(\sgte, \sgte, \sgte) & \dgpred(\sgte,\sgte,\sgte) & \forall n\in\Nat,\;\prompred_n(\sgte) \\
    \hline
  \end{array}
\end{equation*}

\begin{lem}[\LL{} with functorial promotion and digging]\label{lemllfuncinst}
This instance $\superLL(\sset,\depred,\copred,\dgpred,\prompred)$ is \LL{} based on digging and functorial promotion, and it satisfies the cut-elimination axioms and the expansion axiom.
\end{lem}

\begin{proof}
Concerning ($\ocf$), ($\lldg$) and ($\llde$), we have a one-to-one correspondence between the rules of the two systems.
Concerning contraction, the ($\llwk$) and ($\llco$) are exactly cases $k=0$ and $k=2$ of the ($\co{}$) rule.
%
As already remarked in Section~\ref{secceax}, the cut-elimination axioms are satisfied, and the same for the expansion axiom, since $\prompred$ is full.
\end{proof}

\subsection{\ELL}

We consider the instance of \superLL{} given by:
\begin{equation*}
  \begin{array}{|c|c|c|c|c|}
    \hline
    \sset & \depred & \copred & \dgpred & \prompred \\
    \hline
    \{\sgte\} & & \copred_0(\sgte) \qquad \copred_2(\sgte, \sgte, \sgte) & & \forall n\in\Nat,\;\prompred_n(\sgte) \\
    \hline
  \end{array}
\end{equation*}
$\depred$ and $\dgpred$ are the empty (always $\false$) relations. $(\prompred_n)_{n\in\Nat}$ are full.

\begin{lem}
This instance of \superLL{} satisfies the cut-elimination axioms and the expansion axiom. Using notations ${\oc}:={\oc_\sgte}$ and ${\wn}:={\wn_\sgte}$ this instance of \superLL{} is exactly \ELL.
\end{lem}
\begin{proof}
The rules of this instance are exactly the rules of \ELL:
\begin{align*}
\AIC{\vdash\A}
\RL{\llwk}
\UIC{\vdash\wn A,\A}
\DP
&\qquad \leftrightsquigarrow \qquad
\AIC{\vdash\A}
\ZIC{\copred_0(\sgte)}
\RL{\co{}}
\BIC{\vdash\A,\wn_\sgte A}
\DP
\\[2ex]
\AIC{\vdash\wn A,\wn A,\A}
\RL{\llco}
\UIC{\vdash\wn A,\A}
\DP
&\qquad \leftrightsquigarrow \qquad
\AIC{\vdash\wn_\sgte A,\wn_\sgte A,\A}
\ZIC{\copred_2(\sgte,\sgte,\sgte)}
\RL{\co{}}
\BIC{\vdash\wn_\sgte A,\A}
\DP
\\[2ex]
\AIC{\vdash A, A_1, \dotsc, A_n}
\RL{\ocf}
\UIC{\vdash\oc A,\wn A_1, \dotsc, \wn A_n}
\DP
&\qquad \leftrightsquigarrow \qquad
\AIC{\vdash A, A_1, \dotsc, A_n}
\ZIC{\prompred_n(\sgte)}
\RL{\prompred}
\BIC{\vdash\oc_\sgte A,\wn_\sgte A_1, \dotsc, \wn_\sgte A_n}
\DP
\qedhere
\end{align*}
\end{proof}

\subsection{\SLL}

We consider the instance of \superLL{} given by:
\begin{equation*}
  \begin{array}{|c|c|c|c|c|}
    \hline
    \sset & \depred & \copred & \dgpred & \prompred \\
    \hline
    \{\sgte,\epg\} & \depred(\epg) & \forall k\in\Nat,\;\copred_k(\epg,\dotsc,\epg,\sgte) & & \forall n\in\Nat,\;\prompred_n(\sgte) \qquad \forall n\in\Nat,\;\prompred_n(\epg) \\
    \hline
  \end{array}
\end{equation*}
This is a rather non-standard presentation of \SLL.
However using notations $\oc A := \oc_\sgte A$, $\wn A := \wn_\sgte A$, $\flat A := \wn_\epg A$ and $\sharp A := \oc_\epg A$ draws a bridge with presentations inspired by the proof-net syntax, as we can find in the literature~\cite{stratll}.

\begin{lem}[Properties]\label{lemsllprop}
This instance of \superLL{} satisfies the cut-elimination axioms and the expansion axiom.
\end{lem}


\begin{lem}[\SLL{} to \superLL]\label{lemsllsuper}
If we translate ${\oc}\mapsto{\oc_\sgte}$ and ${\wn}\mapsto{\wn_\sgte}$, we can translate
  proofs (resp.\ cut-free proofs) of \SLL{} into
  proofs (resp.\ cut-free proofs) of $\superLL(\sset, \depred, \copred, \dgpred, \prompred)$.
\end{lem}

\begin{proof}
  \begin{align*}
    \AIC{\vdash\overbrace{A,\dotsc,A}^k,\A}
    \RL{\mpx{k}}
    \UIC{\vdash\wn A,\A}
    \DP
    &\qquad\transfo\qquad
    \AIC{\vdash\overbrace{A,\dotsc,A}^k,\A}
    \ZIC{\depred(\epg)}
    \RL{\de}
    \doubleLine
    \BIC{\vdash\wn_\epg A,\dotsc,\wn_\epg A,\A}
    \ZIC{\copred_k(\epg,\dotsc,\epg,\sgte)}
    \RL{\co{}}
    \BIC{\vdash\wn_\sgte A,\A}
    \DP
    \\[2ex]
    \AIC{\vdash A,A_1,\dotsc,A_n}
    \RL{\ocf}
    \UIC{\vdash\oc A,\wn A_1,\dotsc,\wn A_n}
    \DP
    &\qquad\transfo\qquad
    \AIC{\vdash A,A_1,\dotsc,A_n}
    \ZIC{\prompred_n(\sgte)}
    \RL{\prom}
    \BIC{\vdash\oc_\sgte A,\wn_\sgte A_1,\dotsc,\wn_\sgte A_n}
    \DP
      \qedhere
  \end{align*}
\end{proof}

\begin{lem}[\superLL{} to \SLL]\label{lemsupersll}
If we translate ${\oc_\sgte}\mapsto{\oc}$, ${\wn_\sgte}\mapsto{\wn}$, ${\oc_\epg}\mapsto\emptyset$, and ${\wn_\epg}\mapsto\emptyset$ (\ie{} we erase all $\oc_\epg$ and $\wn_\epg$), we can translate proofs (resp.\ cut-free proofs) of $\superLL(\sset, \depred, \copred, \dgpred, \prompred)$ into proofs (resp.\ cut-free proofs) of \SLL.
\end{lem}

\begin{proof}
\allowdisplaybreaks
  \begin{align*}
    \AIC{\vdash A,\A}
    \ZIC{\depred(\epg)}
    \RL{\de}
    \BIC{\vdash\wn_\epg A,\A}
    \DP
    &\qquad\transfo\qquad
    \AIC{\vdash A,\A}
    \DP
    \\[2ex]
    \AIC{\vdash\overbrace{\wn_\epg A,\dotsc,\wn_\epg A}^k,\A}
    \ZIC{\copred_k(\epg,\dotsc,\epg,\sgte)}
    \RL{\co{}}
    \BIC{\vdash\wn_\sgte A,\A}
    \DP
    &\qquad\transfo\qquad
    \AIC{\vdash\overbrace{A,\dotsc,A}^k,\A}
    \RL{\mpx{k}}
    \UIC{\vdash\wn A,\A}
    \DP
    \\[2ex]
    \AIC{\vdash A,A_1,\dotsc,A_n}
    \ZIC{\prompred_n(\sgte)}
    \RL{\prom}
    \BIC{\vdash\oc_\sgte A,\wn_\sgte A_1,\dotsc,\wn_\sgte A_n}
    \DP
    &\qquad\transfo\qquad
    \AIC{\vdash A,A_1,\dotsc,A_n}
    \RL{\ocf}
    \UIC{\vdash\oc A,\wn A_1,\dotsc,\wn A_n}
    \DP
    \\[2ex]
    \AIC{\vdash A,A_1,\dotsc,A_n}
    \ZIC{\prompred_n(\epg)}
    \RL{\prom}
    \BIC{\vdash\oc_\epg A,\wn_\epg A_1,\dotsc,\wn_\epg A_n}
    \DP
    &\qquad\transfo\qquad
    \AIC{\vdash A,A_1,\dotsc,A_n}
    \DP
      \qedhere
  \end{align*}
\end{proof}

\begin{prop}[Cut Elimination for \SLL]
  Cut elimination holds for \SLL.
\end{prop}

\begin{proof}
  We apply Lemma~\ref{lemsllsuper}, Theorem~\ref{thmce} (using Lemma~\ref{lemsllprop}), and Lemma~\ref{lemsupersll}.
\end{proof}

\subsection{\LL}

We consider the following instance:
\begin{equation*}
  \begin{array}{|c|c|c|c|c|}
    \hline
    \sset & \depred & \copred & \dgpred & \prompred \\
    \hline
    \{\sgte\} & \depred(\sgte) & \forall k\in\Nat,\;\copred_k(\sgte,\dotsc,\sgte,\sgte) & \dgpred(\sgte,\sgte,\sgte) & \forall n\in\Nat,\;\prompred_n(\sgte) \\
    \hline
  \end{array}
\end{equation*}
All relations are the full (\ie{} always true) relations.
This makes axioms easy to check (in particular the cut-elimination axioms and the expansion axiom).
As mentioned in Remark~\ref{remco}, we could also restrict to $\copred_k(\sgte,\dotsc,\sgte,\sgte)=\true$ only for $k=0$ and $k=2$, it would not modify the expressiveness of the system. However the Girardization axioms of Table~\ref{tabgir} would not hold.

\begin{lem}[\LL]
The instance, with $\sset=\{\sgte\}$ and full relations, satisfies the Girardization axioms and the induced instance of \superLL{} is equivalent to \LL.
\end{lem}

\begin{proof}
From $\superLL(\sset, \depred, \copred, \dgpred, \prompred)$ to \LL,
since relations are full, the axioms are easily satisfied and we can apply Proposition~\ref{propgir}.
We conclude as in Remark~\ref{remco} for the contraction rules.

From \LL{} to $\superLL(\sset, \depred, \copred, \dgpred, \prompred)$, we use:
\begin{equation*}
\AIC{\vdash A, \wn A_1, \dotsc, \wn A_n}
\RL{\oc}
\UIC{\vdash \oc A, \wn A_1, \dotsc, \wn A_n}
\DP
\qquad\transfo\qquad
\AIC{\vdash A, \wn_\sgte A_1, \dotsc, \wn_\sgte A_n}
\ZIC{\prompred_n(\sgte)}
\RL{\prom}
\BIC{\vdash \oc_\sgte A, \wn_\sgte\wn_\sgte A_1, \dotsc, \wn_\sgte\wn_\sgte A_n}
\ZIC{\dgpred(\sgte,\sgte,\sgte)}
\doubleLine
\RL{\dg}
\BIC{\vdash \oc_\sgte A, \wn_\sgte A_1, \dotsc, \wn_\sgte A_n}
\DP
\qedhere
\end{equation*}
\end{proof}

\subsection{\LLL}

We consider the instance of \superLL{} given by:
\begin{equation*}
  \begin{array}{|c|c|c|c|c|}
    \hline
    \sset & \depred & \copred & \dgpred & \prompred \\
    \hline
    \{\sgte,\epg\}  & & \copred_0(\sgte) & & \prompred_1(\sgte) \\
    & & \copred_1(\sgte,\sgte) \quad \copred_1(\epg,\epg) \quad \copred_1(\epg,\sgte) & & \forall n\in\Nat,\;\prompred_n(\epg) \\
    & & \copred_2(\sgte, \sgte, \sgte) & & \\
    \hline
  \end{array}
\end{equation*}
A key point is $\copred_1(\sgte,\epg)=\false$.

\begin{lem}[Properties]\label{lemlllprop}
This instance of \superLL{} satisfies the cut-elimination axioms, the expansion axiom and the subsumption axioms.
\end{lem}

\begin{proof}
The cut-elimination axioms come easily.
Axiom~(\ref{axsubax}) is immediate.
Axioms~(\ref{axsubrefl}) and~(\ref{axsubtrans}) are satisfied since $\copred_1(\_,\_)$ is an order relation.
Axioms~(\ref{axsubde}) and~(\ref{axsubdg}) are satisfied because \depred{} and \dgpred{} are empty.
Axiom~(\ref{axsubco}) is satisfied since $\copred_k(\e_1,\dotsc,\e_k,\e)=\true$ entails $\e_1=\dotsb=\e_k=\e=\sgte$ or $k=1$ (in which case (\ref{axsubco}) is trivial).
\end{proof}

\begin{lem}[\LLL{} to \superLL]\label{lemlllsuper}
If we translate ${\oc}\mapsto{\oc_\sgte}, {\wn}\mapsto{\wn_\sgte}, {\pg} \mapsto {\oc_{\epg}} $ and $ {\copg} \mapsto {\wn_{\epg}} $, we can translate proofs (resp.\ cut-free proofs) of \LLL{} into proofs (resp.\ cut-free proofs) of $\superLL(\sset, \depred, \copred, \dgpred, \prompred)$.
\end{lem}

\begin{proof}
\allowdisplaybreaks
  \begin{align*}
    \AIC{\vdash \A}
	\RL{\llwk}
	\UIC{\vdash \wn A, \A}
	\DP
	&\qquad \transfo \qquad
	\AIC{\vdash \A}
	\ZIC{\copred_0(\sgte)}
	\RL{\co{}}
	\BIC{\vdash\wn_\sgte A,\A}
	\DP
    \\[2ex]
    \AIC{\vdash \wn A, \wn A, \A}
	\RL{\llco}
	\UIC{\vdash \wn A, \A}
	\DP
	&\qquad \transfo \qquad
	\AIC{\vdash \wn_{\sgte} A, \wn_{\sgte} A, \A}
	\ZIC{\copred_2(\sgte,\sgte,\sgte)}
	\RL{\co{}}
	\BIC{\vdash \wn_{\sgte} A, \A}
	\DP
	\\[2ex]
    \AIC{\vdash A , B}
    \RL{\ocu}
    \UIC{\vdash \oc A , \wn B}
    \DP
    &\qquad\transfo\qquad
    \AIC{\vdash A, B}
    \ZIC{\prompred_1(\sgte)}
    \RL{\prom}
    \BIC{\vdash \oc_{\sgte} A , \wn_{\sgte} B}
    \DP
    \\[2ex]
    \AIC{\vdash A , A_1,\dotsc, A_n,B_1,\dotsc,B_m}
    \RL{\pg}
    \UIC{\vdash \pg A , \copg A_1,\dotsc,\copg A_n, \wn B_1, \dotsc, \wn B_m}
    \DP
    &\qquad\transfo\qquad\\
    & \quad\qquad
    \AIC{\vdash A,A_1,\dotsc,A_n,B_1,\dotsc,B_m}
    \ZIC{\prompred_{n+m}(\epg)}
    \BIC{\vdash \oc_{\epg} A , \wn_{\epg} A_1,\dotsc, \wn_{\epg} A_n, \wn_{\epg} B_1,\dotsc, \wn_{\epg} B_m}
    \ZIC{\copred_1(\epg, \sgte)}
    \doubleLine
    \RL{\co{}}
    \BIC{\vdash \oc_\epg A , \wn_\epg A_1,\dotsc, \wn_\epg A_n, \wn_\sgte B_1, \dotsc,\wn_\sgte B_m}
    \DP
  \end{align*}
      \qedhere
\end{proof}

\begin{lem}[\superLL{} to \LLL]\label{lemsuperlll}
If we translate ${\oc_\sgte}\mapsto{\oc}$, ${\wn_\sgte}\mapsto{\wn}$, ${\oc_\epg}\mapsto\pg$, and ${\wn_\epg}\mapsto\copg$, we can translate proofs (resp.\ cut-free proofs) of $\superLL(\sset, \depred, \copred, \dgpred, \prompred)$ into proofs (resp.\ cut-free proofs) of \LLL.
\end{lem}

\begin{proof}
To prove this result we use Proposition~\ref{subselim} with Lemma~\ref{lemlllprop}.
Then from a proof containing only the ordered promotion rule (and no subsumption rule), we can deduce our translation:
  \begin{align*}
    \AIC{\vdash A, B}
    \ZIC{\sgte \leq \sgte}
    \ZIC{\prompred_1(\sgte)}
    \RL{\promleq}
    \TIC{\vdash \oc_{\sgte} A, \wn_{\sgte} B}
    \DP
    &\qquad\transfo\qquad
    \AIC{\vdash A, B}
    \RL{\ocu}
    \UIC{\vdash \oc A, \wn B}
    \DP
    \\[2ex]
    \AIC{\vdash A, A_1, \dotsc, A_n,B_1,\dotsc,B_m}
    \ZIC{\epg\leq\epg}
    \ZIC{\epg\leq\sgte}
    \ZIC{\prompred_{n+m}(\epg)}
    \RL{\promleq}
    \QIC{\vdash\oc_\epg A,\wn_\epg A_1,\dotsc,\wn_\epg A_n, \wn_\sgte B_1, \dotsc, \wn_\sgte B_m}
    \DP
    &\qquad\transfo\qquad\\ & 
    \AIC{\vdash A, A_1, \dotsc, A_n,B_1,\dotsc,B_m}
    \RL{\pg}
    \UIC{\vdash \pg A, \copg A_1, \dotsc, \copg A_n, \wn B_1, \dotsc, \wn B_m}
    \DP
  \end{align*}
  For the other rules we can refer to \ELL.
\end{proof}

\subsection{Shifting Operators}

We consider the instance given by:
\begin{equation*}
  \begin{array}{|c|c|c|c|c|}
    \hline
    \sset & \depred & \copred & \dgpred & \prompred \\
    \hline
    \{\sgte,\epg\} & \depred(\sgte) & \copred_0(\sgte) \quad \copred_1(\sgte,\sgte) \quad \copred_2(\sgte,\sgte,\sgte) & \dgpred(\sgte,\sgte,\sgte) & \forall n\in\Nat,\;\prompred_n(\sgte) \\
    & \depred(\epg) & \;\copred_1(\epg,\epg)\qquad & \dgpred(\epg,\epg,\epg) & \forall n\in\Nat,\;\prompred_n(\epg) \\
    \hline
  \end{array}
\end{equation*}

\begin{lem}[\LL{} with shifting operators]
This instance is equivalent to \LL{} with shifting operators and satisfies the cut-elimination axioms, the expansion axiom and the Girardization axioms.
\end{lem}

\begin{proof}
Girardization axioms are satisfied because signatures \sgte{} and \epg{} do not interact.
We can apply Proposition~\ref{propgir}.
Then we consider the following correspondence:
\begin{align*}
\AIC{\vdash A, \A}
\ZIC{\depred(\epg)}
\RL{\depred}
\BIC{\vdash \wn_{\epg} A, \A}
\DP
&\qquad \leftrightsquigarrow \qquad
\AIC{\vdash A, \A}
\RL{\shneg}
\UIC{\vdash \shneg A, \A}
\DP
\\[2ex]
\AIC{\vdash A, \wn_{\epg} A_1, \dotsc, \wn_{\epg} A_n}
\ZIC{\dgpred(\epg, \epg, \epg)}
\ZIC{\prompred_n(\epg)}
\RL{\promg}
\TIC{\vdash \oc_{\epg} A, \wn_{\epg} A_1, \dotsc, \wn_{\epg} A_n}
\DP
&\qquad \leftrightsquigarrow \qquad
\AIC{\vdash A, \shneg A_1, \dotsc, \shneg A_n}
\RL{\shpos}
\UIC{\vdash \shpos A, \shneg A_1, \dotsc, \shneg A_n}
\DP
\qedhere
\end{align*}
\end{proof}

\subsection{\seLL}

An instance of \seLL{} is determined by: a pre-ordered set $(\sset,{\leqse})$, and two subsets $\sset_W$ and $\sset_C$ of \sset{} which are upward closed with respect to $\leqse$.
From these data, we can define an associated instance of \superLL{} built on the same set of exponential signatures by considering:
\begin{equation*}
  \begin{array}{|c|c|c|c|c|}
    \hline
    \sset & \depred & \copred & \dgpred & \prompred \\
    \hline
    \sset & \depred(\e) & \copred_0(\e)\text{ if }\e\in\sset_W & \dgpred(\e,\e',\e')\text{ if }\e\leqse\e' & \forall n\in\Nat,\;\prompred_n(\e) \\
    & & \copred_1(\e,\e) & & \\
    & & \copred_2(\e,\e,\e)\text{ if }\e\in\sset_C & & \\
    \hline
  \end{array}
\end{equation*}
All exponential signatures are universally quantified: $\depred(\e)$ above, for example, means $\forall\e\in\sset,\;\depred(\e)$.

\begin{lem}[Properties]\label{lemseLLprop}
$\superLL(\sset, \depred, \copred, \dgpred, \prompred)$ satisfies the cut-elimination axioms, the expansion axiom and the Girardization axioms.
\end{lem}

\begin{proof}
%
Concerning the Girardization axioms, the key property is the definition of $\dgpred$: $\dgpred(\e_1,\e_2,\e_3)\iff \e_1\leqse\e_2 \wedge \e_2=\e_3$.
Let us focus on~(\ref{axgirco}).
For $k=1$, we choose $\eb'_1:=\e_2$.
For $k=0$ and $k=2$, we rely on the upward closure of $\sset_W$ and $\sset_C$ (by taking $\eb'_1,\eb'_2:=\e_2$ for $k=2$).
\end{proof}

\begin{lem}[\seLL{} to \superLL]\label{lemseLLsuper}
We can translate proofs (resp.\ cut-free proofs) of $\seLL(\sset, {\leqse}, \sset_W, \sset_C)$ into proofs (resp.\ cut-free proofs) of $\superLL(\sset, \depred, \copred, \dgpred, \prompred)$.
\end{lem}

\begin{proof}
We can apply the following translations:
\begin{align*}
\AIC{\vdash A, \A}
\RL{\llde[\e]}
\UIC{\vdash \wn_\e A, \A}
\DP
&\qquad \transfo \qquad
\AIC{\vdash A, \A}
\ZIC{\depred(\e)}
\RL{\de}
\BIC{\vdash \wn_\e A, \A}
\DP
\\[2ex]
\AIC{\vdash \A}
\AIC{\e\in\sset_W}
\RL{\llwk[\e]}
\BIC{\vdash \wn_\e A, \A}
\DP
&\qquad \transfo \qquad
\AIC{\vdash \A}
\AIC{\e\in\sset_W}
\UIC{\copred_0(\e)}
\RL{\co{}}
\BIC{\vdash \wn_\e A, \A}
\DP
\\[2ex]
\AIC{\vdash \wn_\e A, \wn_\e A, \A}
\AIC{\e\in\sset_C}
\RL{\llco[\e]}
\BIC{\vdash \wn_\e A, \A}
\DP
&\qquad \transfo \qquad
\AIC{\vdash \wn_\e A, \wn_\e A, \A}
\AIC{\e\in\sset_C}
\UIC{\copred_2(\e,\e,\e)}
\RL{\co{}}
\BIC{\vdash \wn_\e A, \A}
\DP
\\[2ex]
\AIC{\vdash A, \wn_{\e_1}A_1, \dotsc, \wn_{\e_n}A_n}
\AIC{\e\leqse\e_1\quad\dotsb\quad\e\leqse\e_n}
\RL{\oc_\e}
\BIC{\vdash \oc_\e A, \wn_{\e_1}A_1, \dotsc, \wn_{\e_n}A_n}
\DP
&\qquad \transfo \qquad \\ & \hspace{-6pt}
\AIC{\vdash A, \wn_{\e_1}A_1, \dotsc, \wn_{\e_n}A_n}
\ZIC{\prompred_n(\e)}
\RL{\prom}
\BIC{\vdash \oc_\e A, \wn_\e\wn_{\e_1}A_1, \dotsc, \wn_\e\wn_{\e_n}A_n}
\AIC{1\leq i\leq n}
\noLine
\UIC{\e\leqse\e_i}
\UIC{\dgpred(\e,\e_i,\e_i)}
\doubleLine
\RL{\dg}
\BIC{\vdash \oc_\e A, \wn_{\e_1}A_1, \dotsc,\wn_{\e_n}A_n}
\DP
\end{align*}
\end{proof}

\begin{lem}[\superLL{} to \seLL]\label{lemsuperseLL}
We can translate proofs (resp.\ cut-free proofs) of $\superLL(\sset, \depred, \copred, \dgpred, \prompred)$ into proofs (resp.\ cut-free proofs) of $\seLL(\sset, {\leqse}, \sset_W, \sset_C)$.
\end{lem}

\begin{proof}
By Lemma~\ref{lemseLLprop} and Proposition~\ref{propgir}, we can translate the proofs of the current instance of \superLL{} into proofs without digging and functorial promotion but with Girard's promotion instead.
Such proofs correspond to \seLL{} proofs since we have:
\begin{equation*}
\AIC{\vdash A, \wn_{\e_1} A_1, \dotsc, \wn_{\e_n} A_n}
\AIC{1\leq i \leq n}
\noLine
\UIC{\e \leqse \e_i}
\UIC{\dg(\e, \e_i, \e_i)}
\ZIC{\prompred_n(\e)}
\RL{\promg}
\TIC{\vdash \oc_{\e} A, \wn_{\e_1} A_1, \dotsc, \wn_{\e_n} A_n}
\DP
\qquad \transfo \qquad
\AIC{\vdash A, \wn_{\e_1} A_1, \dotsc, \wn_{\e_n} A_n}
\AIC{1\leq i \leq n}
\noLine
\UIC{\e \leqse \e_i}
\RL{\oc_{\e}}
\BIC{\vdash \oc_{\e} A, \wn_{\e_1} A_1, \dotsc, \wn_{\e_n} A_n}
\DP
\qedhere
\end{equation*}
\end{proof}

\subsection{\texorpdfstring{\BsLL}{BSLL}}

We consider an ordered semi-ring $(\sset,{+},0,{\cdot},1,{\leqse})$.
From it we can define an instance of \superLL:
\begin{equation*}
  \begin{array}{|c|c|c|c|c|}
    \hline
    \sset & \depred & \copred & \dgpred & \prompred \\
    \hline
    \sset & \depred(1) & \copred_0(0) & \dgpred(\e_1,\e_2,\e_1\cdot\e_2) & \forall n\in\Nat,\;\prompred_n(\e) \\
    & & \copred_1(\e,\e')\text{ if }\e\leqse\e' & & \\
    & & \copred_2(\e_1,\e_2,\e_1+\e_2)& & \\
    \hline
  \end{array}
\end{equation*}

\begin{lem}[Properties]\label{lemBsLLprop}
$\superLL(\sset,\depred,\copred,\dgpred,\prompred)$ satisfies the cut-elimination axioms, the expansion axiom and the Girardization axioms.
\end{lem}

\begin{proof}
Concerning the Girardization axioms, we mostly rely on Remark~\ref{remgirax}.
\end{proof}

\begin{lem}[\BsLL{} to \superLL]\label{lemBsLLsuper}
We can translate proofs (resp.\ cut-free proofs) of $\BsLL(\sset,{+},0,{\cdot},1,{\leqse})$ into proofs (resp.\ cut-free proofs) of $\superLL(\sset, \depred, \copred, \dgpred, \prompred)$.
\end{lem}

\begin{proof}
We only give the translation for the promotion rule:
\begin{equation*}
\AIC{\vdash A, \wn_{\e_1}A_1, \dotsc, \wn_{\e_n}A_n}
\RL{\oc_{\_\cdot\_}}
\UIC{\vdash \oc_\e A, \wn_{\e\cdot\e_1}A_1, \dotsc, \wn_{\e\cdot\e_n}A_n}
\DP
\qquad \transfo \qquad
\AIC{\vdash A, \wn_{\e_1}A_1, \dotsc, \wn_{\e_n}A_n}
\ZIC{\prompred_n(\e)}
\RL{\prom}
\BIC{\vdash \oc_\e A, \wn_\e\wn_{\e_1}A_1, \dotsc, \wn_\e\wn_{\e_n}A_n}
\AIC{1\leq i\leq n}
\UIC{\dgpred(\e,\e_i,\e\cdot\e_i)}
\doubleLine
\RL{\dgpred}
\BIC{\vdash \oc_\e A, \wn_{\e\cdot\e_1}A_1, \dotsc, \wn_{\e\cdot\e_n}A_n}
\DP
\qedhere
\end{equation*}
\end{proof}

\begin{lem}[\superLL{} to \BsLL]\label{lemsuperBsLL}
We can translate proofs (resp.\ cut-free proofs) of $\superLL(\sset, \depred, \copred, \dgpred, \prompred)$ into proofs (resp.\ cut-free proofs) of $\BsLL(\sset,{+},0,{\cdot},1,{\leqse})$.
\end{lem}

\begin{proof}
By Lemma~\ref{lemBsLLprop} and Proposition~\ref{propgir}, we can translate all the proof of our instance of \superLL{} into proofs without digging and functorial promotion but with Girard's promotion instead, which in this case is exactly the promotion in \BsLL.
\end{proof}

\section{Conclusion}

We have presented \superLL{}, a parameterized extension of linear logic.
We have shown that, under some conditions, this system eliminates cuts (Theorem \ref{thmce}).
We have described many existing linear logic systems as instances of \superLL{} (Section \ref{secsubsyst}), so that cut elimination for these systems can be easily deduced.

Our general goal is not only to prove these theorems on paper, but also to formally prove them on a proof assistant.
In this context, it is particularly interesting to be able to factorize the code of many proofs into one.
This is still work in progress, but the objective is to use \superLL{} as new core system for the \Coq{} library \yalla~\cite{yalla}.
This would also allow users to design their own linear logic variant as an instance of \superLL{} and to rely on the provided cut-elimination proof.

However, not every linear logic system is an instance of \superLL.
For instance, Bounded Linear Logic (\BLL)~\cite{bll} is a system where signatures are polynomials with dependencies inside formulas.
Other systems constrain the exponential rules by global restrictions in the proofs which are not captured by \superLL{} (see for example \logsys{L$^3$} and \logsys{L$^4$}~\cite{l3,stratll}).

The work presented here focuses on the sequent calculus presentation of linear systems.
However a key syntactic contribution of Linear Logic is the introduction of the graphical syntax of proof-nets~\cite{ll}.
Defining a notion of proof-nets for \superLL{} should not be too difficult since the cut-elimination steps we deal with in the sequent calculus should be local enough.
It would be an important step towards the study of strong normalization for \superLL~\cite{snll}.

\bibliographystyle{eptcs}
\bibliography{superLL}
\end{document}